\documentclass[twoside,11pt]{article}
\usepackage{jair, theapa, rawfonts}
\usepackage{graphicx}
\usepackage{amssymb, amsmath}

 \usepackage{relsize}
 \usepackage[makeroom]{cancel}

\newtheorem{theorem}{Theorem}
\newtheorem{lemma}{Lemma}
\newtheorem{observation}{Observation}
\newtheorem{proposition}{Proposition}

\newtheorem{definition}{Definition}
\newtheorem{example}{Example}

\newenvironment{proof}[1][Proof]{\vspace{\baselineskip}\noindent\textbf{#1:} }{\rule{0.5em}{0.67em}}

\usepackage{xparse}
\usepackage{xifthen}
\newcommand{\citet}[1]{\citeA{#1}}
\newcommand{\citealp}[1]{\citeR{#1}}
\DeclareDocumentCommand{\citep}{ O{} O{} m }{%
	\ifthenelse%
		{\isempty{#1}}%
		{\ifthenelse%
			{\isempty{#2}}%
			{\cite{#3}}%
			{\cite[#2]{#3}}%
		}%
		{\ifthenelse%
			{\isempty{#2}}%
			{\cite<#1>{#3}}%
			{\cite<#1>[#2]{#3}}%
		}%
}%

\jairheading{61}{2018}{407-431}{08/17}{03/18}
\ShortHeadings{Cycles and Intractability}
{Zwicker}
\firstpageno{407}

\begin{document}

\title{Cycles and Intractability in a Large Class\\ of Aggregation Rules}

\author{\name William S. Zwicker \email zwickerw@union.edu \\
       \addr Mathematics Department, Union College,\\
       Schenectady, NY 12308 USA
}

\maketitle

\begin{abstract}
  We introduce the $(j,k)$-Kemeny rule -- a generalization of Kemeny's voting rule
that aggregates $j$-chotomous weak orders into a $k$-chotomous weak order. Special cases of $(j,k)$-Kemeny include approval voting, the mean rule and Borda mean rule, as well as the Borda count and plurality voting.  Why, then, is the winner problem computationally tractable for each of these other rules, but intractable for Kemeny?  We show that intractability of winner determination for the $(j,k)$-Kemeny rule first appears at the $j=3$, $k=3$ level. The proof rests on a reduction of \emph{max cut} to a related problem on weighted tournaments, and reveals that computational complexity arises from the cyclic part in the fundamental decomposition 
of a weighted tournament into cyclic and cocyclic components. 
Thus the existence of 
majority cycles -- the engine driving both Arrow's impossibility theorem and the Gibbard-Satterthwaite theorem -- also serves as a source of computational complexity in social choice.
\end{abstract}

\section{Introduction}
\label{introduction}

  In their seminal paper, \citet{BTT} showed that determining the optimal Kemeny ranking in an election is $\mathit{NP}$-hard; \citet{HSV} later showed completeness for $P^{\mathit{NP}}_{||}$.  We introduce the $(j,k)$-Kemeny rule, a  generalization wherein ballots are weak orders with $j$ indifference classes (``$j$-chotomous'' weak orders) and the outcome is a weak order with $k$ indifference classes. Different values of $j$ and $k$ yield rules of interest in social choice theory as special cases, including approval voting, the mean rule and Borda mean rule (see \citealp{DP}; \citealp{DPZ}; \citealp{BP}), the Borda count voting rule, and plurality voting.

Why, then, is winner determination computationally tractable for each of these other rules, but \emph{not} tractable for Kemeny?  We show that these other rules each satisfy $j \leq 2$ or $k \leq 2$, while winner intractability for the $(j,k)$-Kemeny rule first appears at the $j=3$, $k=3$ level. This follows from our central result: the well-known $\mathit{NP}$-complete \emph{max cut problem} for undirected graphs can be polynomially reduced to \emph{max 3-OP} (``OP" for ``Ordered Partition"), a version of \emph{max cut} for weighted directed graphs or tournaments
 in which vertices are partitioned into three pieces rather than two, and pieces are ordered.   In this connection, it's worth remarking that $(3,3)$-Kemeny is a new rule -- one that seems potentially useful in some contexts.  Imagine that a small subcommittee is screening job applicants into three categories: those whose files are strong and should be read by the larger committee, those who are not appropriate for the position, and a middle group of applicants worth further consideration if everyone in the top group turns out to be unavailable.  If each member of the subcommittee privately sorts in this fashion, then $(3,3)$-Kemeny seems to be an appropriate way to aggregate those individual views into a collective recommendation.\footnote{This is a variant of the story told by \citet{DPZ} as motivation for the Mean Rule (which is $(2,2)$-Kemeny).}
 
The proof reveals that computational complexity arises from the cyclic component in the orthogonal decomposition $\overrightarrow{w}=\overrightarrow{w}_{\emph{cycle}} + \overrightarrow{w}_{\emph{cocycle}}$ of the weighted tournament induced by a profile, in which $\overrightarrow{w}_{\emph{cycle}}$ serves as a measure of underlying tendency towards a cycle of majority preference, while $\overrightarrow{w}_{\emph{cocycle}}$ resists that tendency (see \citealp{Young}; \citealp{Z}). 
  Thus majority cycles -- the engines driving both Arrow's impossibility theorem and the Gibbard-Satterthwaite theorem -- also serve as a source of computational complexity in social choice.   

The rest of the paper is organized as follows: Section \ref{Preliminaries} introduces some notation and key concepts, while Section \ref{Decomp} briefly reviews the orthogonal decomposition of a weighted tournament into cyclic and cocyclic components.  Readers to whom the decomposition is new may wish to read the exposition by \citet{DPZ}, which includes additional diagrams and worked examples, along with proofs that the two components are orthogonal under the standard Euclidean inner product and that the corresponding vector spaces are complements in the Euclidean space of all possible edge-weightings of a tournament (so that the decomposition always exists and is unique).

Section \ref{MaxCutSection} dissects the relationship between the standard \emph{max cut} problem for (weighted, undirected) graphs, and our version \emph{max }$k$-\emph{OP} for weighted tournaments.  While standard \emph{max cut} is $\mathit{NP}$-complete for vertex partitions into two or more pieces, the directed version is polynomial-time for partitions into two pieces and also for partitions into more pieces when the cyclic component is zero, aka the \emph{purely acyclic case}. For the general case it is $\mathit{NP}$-complete for partitions into $k \geq 3$ pieces. 
Table 1 summarizes these results, which appear as Theorems 1 (two hardness results) and 2 (two easiness results).  Together, they show that the cyclic component functions as the sole source of intractability in \emph{max }$k$-\emph{OP}; the proofs explicate this critical distinction between partitions into two vs. three pieces.

\begin{table}
\begin{centering}
\begin{tabular}{|l|l|}
\hline 
\bf{\emph{max }$k$-\emph{OP} special case} & \bf{Complexity}  \tabularnewline
\hline 
\hline 

\emph{max }$2$-\emph{OP} & $\mathit{P}$\hspace{19mm} (Theorem 2.2) \tabularnewline

\hline

Purely acyclic case of \emph{max }$k$-\emph{OP}, $k \geq 3$ & $\mathit{P}$\hspace{19mm} (Theorem 2.1) \tabularnewline

\hline

Purely cyclic case of \emph{max }$k$-\emph{OP}, $k \geq 3$ & $\mathit{NP}$-complete (Theorem 1.1)\tabularnewline

\hline

Transitive case of \emph{max }$k$-\emph{OP}, $k \geq 3$ & $\mathit{NP}$-complete (Theorem 1.2)\tabularnewline

\hline

\end{tabular}\protect\caption{Special cases of \emph{max }$k$-\emph{OP}.  
}

\par\end{centering}

\end{table}

We introduce the $(j,k)$-Kemeny rule in Section \ref{jkSection}, and observe that a number of familiar aggregation rules are special cases (see Table 2).
Calculating outcomes for these rules amounts to optimizing cases of  \emph{max }$k$-\emph{OP}; thus we can transfer complexity results from Section \ref{MaxCutSection} to the outcome determination problem for these rules.  In particular, $j \leq 2$ guarantees that the weighted tournament is purely acyclic ($\overrightarrow{w}_{\emph{cycle}}=\mathbf{0}$) while $k \leq 2$ guarantees that $\overrightarrow{w}_{\emph{cycle}}$ plays no role in the aggregation; neither guarantee applies when $j,k\geq3$.  However, if the profile \emph{happens} to be purely acyclic, the winner can be computed in polynomial time for any values of $j$ and $k$; for example, the Kemeny voting rule coincides with the Borda count in this scenario. 

Several ideas developed here (but not those related to complexity) first appeared in work by \citet{DPZ}.  In the concluding Section \ref{Conclusions} we touch on the important, behind-the-scenes role played by notions of generalized scoring rule (implicit in \citealp{Myerson}, and explicit in \citealp{Zwicker}; \citealp{CRX};  \citealp{XC}; \citealp{Z2018}), and on the distinction between the median procedure and the alternative generalization of Kemeny's voting rule used in this paper. Hudry has written several papers considering complexity issues for special cases of the median procedure, including the case of aggregating weak orders (see \citealp{Hudry} and comments in Section \ref{Conclusions}).  

\begin{table}
\begin{centering}
\begin{tabular}{|l|l|l|}
\hline 
\bf{Special case} & \bf{Equivalence to known rule} & \bf{Complexity} \tabularnewline
\hline 
\hline 
$(2,2)$-Kemeny &  Mean Rule & $\mathit{P}$ \tabularnewline
\hline 
$(2,2_1)$-Kemeny & Approval Voting, with winner as outcome & $\mathit{P}$ \tabularnewline

\hline 
$(2,\mathcal{L})$-Kemeny & Approval Voting, with  ranking as outcome & $\mathit{P}$ \tabularnewline

\hline 
$(2_1,2_1)$-Kemeny & Plurality Voting, with  winner as outcome & $\mathit{P}$ \tabularnewline

\hline 
$(2_1,\mathcal{L})$-Kemeny & Plurality Voting, with  ranking as outcome & $\mathit{P}$ \tabularnewline

\hline 
$(\mathcal{L}, 2)$-Kemeny & Borda Mean Rule & $\mathit{P}$ \tabularnewline

\hline 
$(\mathcal{L}, 2_1)$-Kemeny & Borda Voting, with winner as outcome & $\mathit{P}$\tabularnewline

\hline 
$(3, 3)$-Kemeny & (new rule: aggregates trichotomous weak orders) & $\mathit{NP}$-hard \tabularnewline

\hline 
$(\mathcal{L}, \mathcal{L})$-Kemeny & Kemeny Voting, with ranking as outcome & $\mathit{NP}$-hard \tabularnewline

\hline 

\end{tabular}\protect\caption{Special cases of $(j,k)$-Kemeny winner determination.  In an abuse of notation, we allow $j$ or $k$ to take on the value $2_1$ (meaning that the inputs or outputs are \emph{univalent} orders -- dichotomous weak orders whose top set is a singleton) or $\mathcal{L}$ (meaning that  inputs or outputs are \emph{linear} orders).
For voting rules that assign scores to candidates, such as approval voting, the outcome of the election can be expressed in the form of a winner (candidate with highest score) or in the form of a ranking (in decreasing order of score). The Mean Rule outcome ranks all alternatives with above average approval score over all those with below average score and the Borda Mean Rule acts similarly, but with Borda scores replacing approval scores (for details see \citealp{DPZ}; see \citealp{BP} for a recent axiomatic characterization of the Borda Mean Rule). }

\par\end{centering}

\end{table}

\section{Preliminaries} \label{Preliminaries}

A \emph{tournament} $(V, E)$ is a complete, antisymmetric directed graph, wherein $V$ is a finite set of \emph{vertices} and $E\subseteq V\times V$ is a set of directed edges (or \emph{arcs}) satisfying that for each $x,y\in V$ with $x \neq y$ exactly one of the pairs $(x,y)$ or $(y,x)$ belongs to $E$; we depict $(x,y)\in E$ as an arrow $x\rightarrow y$. For the remainder of the paper we will put arrows over symbols for tournaments and their components, to distinguish them from related but distinct objects associated with undirected graphs; for example, in Section \ref{MaxCutSection} the vertex sets $V$ and $\overrightarrow{V}$ are related but different. A \emph{weighted tournament} $\overrightarrow{\mathcal{H}}=(\overrightarrow{V}, \overrightarrow{E},\overrightarrow{w})$ is equipped with an \emph{edge-weight assignment} $\overrightarrow{w}$ of numerical weights to the edges, which can be interpreted as the \emph{flow} of some substance through the channels of a network.   The \emph{antisymmetric extension} of $\overrightarrow{w}$ assigns the value $\overrightarrow{w}^\emph{ext} (b,a) = -\overrightarrow{w}(a,b)$ for each $(b,a)\notin \overrightarrow{E}$.  For example, a $5$ on the $a \rightarrow b$ edge could indicate a flow of $5$ amps of electricity in a wire from $a$ to $b$ (or, as we discuss next, a flow of \emph{net preference} arising from a \emph{profile} of weak order relations).    A weight of \emph{negative} $5$ on $a \rightarrow b$ would tell us that $5$ units are flowing in the opposite direction, from $b$ to $a$, with $\overrightarrow{w}^\emph{ext}(b,a) = +5$.  Thus, we view the initial edge orientation as an arbitrary choice that serves bookkeeping purposes and does not indicate the actual flow direction.  In practice, we will be abusing notation by omitting the \emph{ext} superscript, conflating $\overrightarrow{w}$ with its antisymmetric extension through the rest of the paper.\footnote{\label{TwoEdges}Alternately, one might obviate any need for extending $\overrightarrow{w}$ by including both edges $a \rightarrow b$ and $b \rightarrow a$ to start with; for a weighted tournament one would then demand that the weight assignment $\overrightarrow{w}$ be an antisymmetric function.  We don't take that approach here, as it somewhat complicates some later arguments.  However, a planned sequel \citep{Z2018} requires \emph{weighted bi-tournaments} in which the antisymmetry requirement is dropped, allowing for the decomposition of $\overrightarrow{w}$ into (orthogonal) antisymmetric and symmetric components; the symmetric component can then be interpreted as a weight assignment to undirected edges. 
}   

A \emph{weak order} relation $\geq$ on $\overrightarrow{V}$ is a binary relation that is transitive, reflexive, and complete (total); if it is additionally antisymmetric then it is a \emph{linear order} relation.
The induced equivalence relation $\sim$ is defined by $a \sim b$ when $a \geq b$ and $b \geq a$ both hold. The most common interpretation of $a\geq b$, as an expression of (weak) preference for $a$ over $b$, has shaped some of the accompanying terminology; for example, under this interpretation $a \sim b$ expresses indifference, and so equivalence classes under $ \sim $ have come to be known as \emph{indifference classes}.   We will employ such more-or-less standard terms, but the reader should not allow this terminology to discourage interpretations other than preference.  

The indifference classes partition $\overrightarrow{V}$, and we can identify a weak order $\geq$ with a linear ordering of its indifference classes.  Thus, a \emph{dichotomous weak order}, with exactly two indifference classes, can be expressed in the form $\{T > B\}$, where $T$ is the top class and $B$ is the bottom, while a \emph{trichotomous weak order} can be written as $\{T >M> B\}$ (which, as a set of ordered pairs, is $\{ (x,y) \in \overrightarrow{V} \times \overrightarrow{V} \; |\;  x \in T$, or $y\in B$, or $(x,y) \in M \times M\}$.  More generally, we will say that a weak order is $j$\emph{-chotomous} if it has exactly $j$ indifference classes. The \emph{univalent} orders are the dichotomous weak orders of form  $\{ \{ x \}  > \overrightarrow{V}\setminus \{ x \} \}$, wherein the top set $T$ is a singleton.

Preference is not the only natural interpretation supported by the aggregation rules we have in mind. For example, a high school teacher might use a dichotomous weak order $\{T > B\}$ to recommend that some of her former pupils -- those she places in $T$ --  be put in a more advanced mathematics class next year, while the others would benefit from more review.  Perhaps several former teachers similarly weigh in with placement recommendations for the same set $\overrightarrow{V}$ of students.  In this case, we obtain a \emph{profile} $\Pi = \{\geq_i\}_{i \in N}$ of dichotomous weak orders, where $N$ is the set of teachers.  

More generally we are interested in aggregating a profile $\Pi$ of $j$-chotomous weak orders on a set $\overrightarrow{V}$.    
Such a profile induces a weighted tournament $\overrightarrow{\mathcal{H}}_\Pi$ on $\overrightarrow{V}$, in which the weight 

\[\overrightarrow{w}(a,b) = |\{ i \in N \; | \; a \geq _i b\}| - |\{ i \in N \; | \; b \geq _i a\} |  \]

\noindent on the $a \rightarrow b$ edge represents the numerical margin by which agents $i$ in $N$ with $a \geq  _i b$ outnumber those with $b \geq _i a$.\footnote{Which, if negative, amounts to a positive margin for $b$ over $a$.}   
We will refer to the resulting weighted tournament as the \emph{flow of net preference}.  This is another term arising from the preference interpretation -- specifically, the context in which agents are voters in an election and orders in the profile are ballots in the form of linear preference rankings of the candidates, so that $\overrightarrow{w}(a,b)$ represents the signed margin by which the majority express a preference for candidate $a$ over $b$.

This weighted tournament $\overrightarrow{\mathcal{H}}_\Pi$ provides only partial information about the original profile $\Pi$, and our analysis here applies only to aggregation rules for which that information suffices to compute the aggregated binary relation.\footnote{For example, this information is insufficient to compute election outcomes under \emph{single transferable vote} (aka \emph{instant run-off} or \emph{Hare} voting).}  In the voting context, with linear order ballots, these are known as the ``C2'' rules according to the classification by \citet{F}. Nonetheless, our results apply to some rules -- such as plurality voting -- known to be non-C2, as our profiles are not limited to linear orders.  For example, given a profile of univalent orders $\{ \{ x_i \}  > \overrightarrow{V}\setminus \{ x_i \} \}$, each representing voter $i$'s plurality ballot for candidate $x_i$, the induced weighted tournament \emph{does} encode all differences in plurality scores among the candidates.  

\begin{definition}\label{PairMajRel}
\emph{[Pairwise majority relation]} 
Let $\Pi$ be a profile of weak orders on $\overrightarrow{V}$ and $\overrightarrow{\mathcal{H}}_\Pi =(\overrightarrow{V},\overrightarrow{E}, \overrightarrow{w})$ be the induced weighted tournament.  For $a,b \in \overrightarrow{V}$, we write $a >^\mu b$ if $\overrightarrow{w}(a,b) > 0,$ $a \geq ^\mu b$ if $\overrightarrow{w}(a,b) \geq 0,$ and $a \sim ^\mu b$ if $\overrightarrow{w}(a,b) = 0.$ 
\end{definition}

\noindent  When $\Pi$ is a profile of linear orders, $a >^\mu b$ tells us that a majority of the individual orders in the profile have $a >_i b$, but the term ``pairwise majority relation'' is something of a misnomer in the broader context of weak orders.\footnote{
For example if $\Pi$ is a profile of $100$ weak orders, one with $a >_i b$ and the remaining $99$ with $a \sim_i b$, then $a >^\mu b$ holds. No majority has $a >_i b$, while there exist majorities satisfying \emph{both} $a \geq _i b$ and $b \geq_i a$.
}  

We'll refer to a sequence 
\begin{equation} \label{cycle}
a_1 >^\mu a_2 >^\mu \dots >^\mu a_k >^\mu a_1
\end{equation}
\noindent as a \emph{majority cycle} or \emph{Condorcet cycle}.  Such cycles are ruled out if $>^\mu$ is \emph{transitive}, satisfying  $a >^\mu b >^\mu c \Rightarrow a >^\mu c$ for all $a,b,c,\in \overrightarrow{V}$.  In general, \emph{acyclicity} of $>^\mu$ is a strictly weaker property (of the underlying weighted tournament or profile) than is transitivity of $>^\mu$, or the even stronger property of transitivity of $\geq^\mu$. However, under the conditions that prevail in Section \ref{MaxCutSection}, all ordinal forms of transitivity and acyclicity are equivalent.\footnote{Without   such conditions, the absence of cycles like that of line (\ref{cycle}) does not preclude weaker forms such as $a_1 >^\mu a_2  >^\mu a_3 \sim^\mu a_1$, which violate transitivity of $>^\mu$.  There exist multiple inequivalent (ordinal) forms of acyclicity and of transitivity.  However, when $\geq ^\mu$ is antisymmetric, or when $\overrightarrow{w}$ is purely cyclic, these distinctions vanish -- and one or the other of these conditions hold for the weighted tournaments constructed in the reductions used to prove Theorems \ref{Thm1}.1 and \ref{Thm1}.2.}

\section{The Orthogonal Decomposition of a Weighted Tournament} \label{Decomp}

A weighted tournament $\overrightarrow{\mathcal{H}}$ that assigns weight $1$ to each of the edges in a cycle\footnote{Recall that we have extended $\overrightarrow{w}$ to assign weights to edges in both directions.  Thus it may be that the original (unextended) edge-weighting assigned \emph{negative} $1$ to some edge $a_j \leftarrow a_{j+1}$ initially oriented oppositely to the sense of the cycle, but in this case the extended version indeed assigns weight $+1$ to the $a_j \rightarrow a_{j+1}$ edge.}
\[
a_1 \rightarrow a_2 \rightarrow \dots \rightarrow a_k \rightarrow a_1
\]
\noindent and weight $0$ to all other edges, is called a \emph{basic cycle} (or \emph{loop current} in electrical engineering speak); the \emph{basic cocycle}\footnote{Also called coboundary (see \citealp{Harary}).} (or \emph{source}) $a^\star$ \emph{at vertex} $a$ assigns weight $1$ to each edge  $a \rightarrow x$ from $a$ to another vertex, weight \emph{negative} $1$ to each edge  $a \leftarrow x$, and weight $0$ to each edge not incident to $a$.  Figure 1 shows examples for a tournament on four vertices.

\begin{figure}[!ht]
{\centering
\includegraphics[height=32mm,width=80mm] {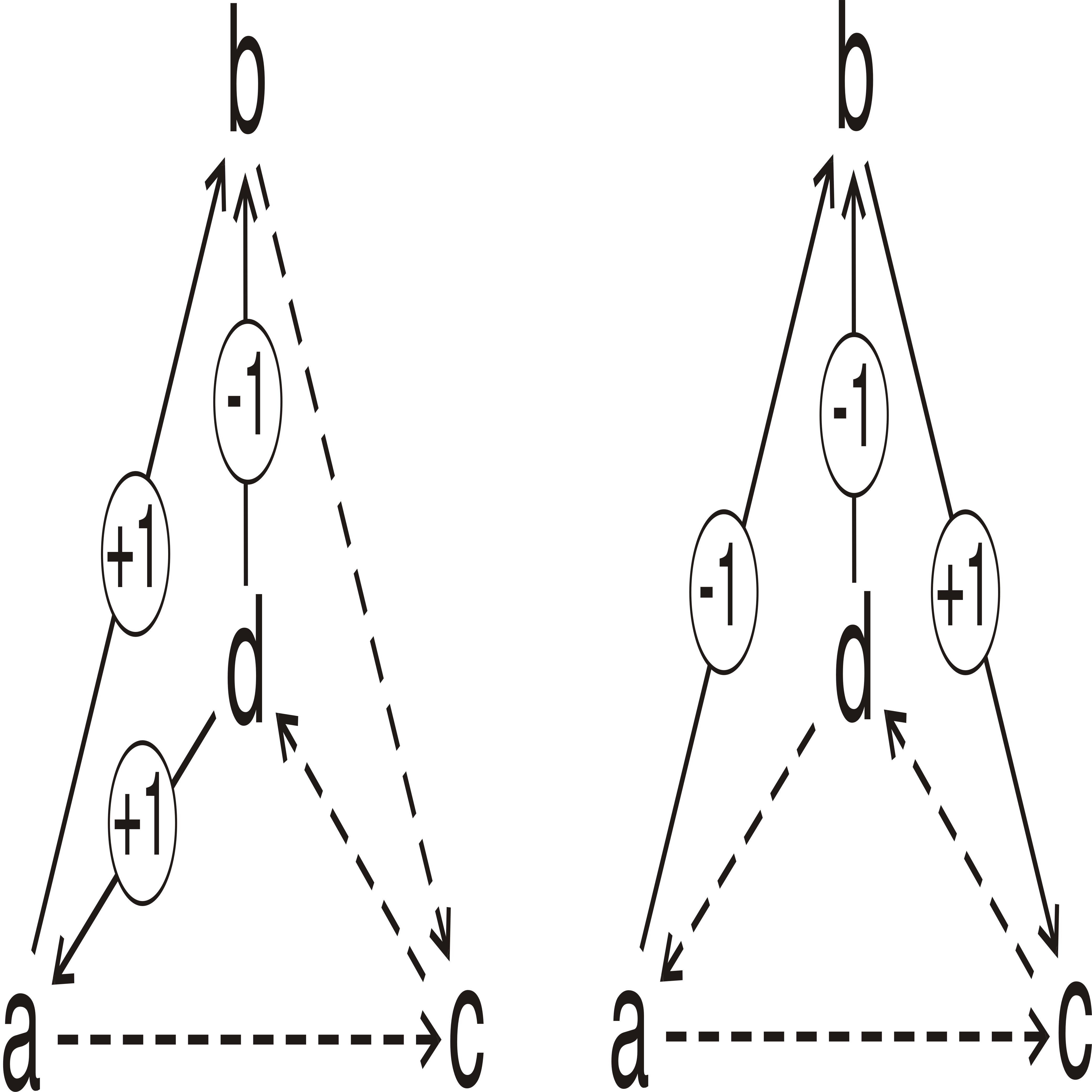}
\caption{A basic cycle (on the left) and the basic cocycle $b^\star$ (on the right).  A dashed edge indicates an edge-weight of zero.}\label{FigZero}}
\end{figure}

Any fixed enumeration of the set $ \overrightarrow{E}$ of edges of a tournament will identify each edge-weight assignment $\overrightarrow{w}$ with a point\footnote{The $i^{th}$ coordinate of $\overrightarrow{w}$ is just the weight assigned, by $\overrightarrow{w}$, to the $i^{th}$ edge in the enumeration.}  in the vector space $\mathbb{R}^{\frac{m(m-1)}{2}}$ (where $m = |\overrightarrow{V}|$) endowed with the standard inner product (for $\mathbf{u} = (u_1 ,u_2, \dots),$ $\mathbf{v} = (v_1 ,v_2, \dots)$, $\mathbf{u} \cdot \mathbf{v} = u_1 v_1 + u_2 v_2 + \dots$ -- equivalently, multiply weights on corresponding edges, then add these products\footnote{This requires that the underlying tournaments supporting $\mathbf{u} $ and $\mathbf{v}$ have matching edge orientations (which is why we cannot freely reverse the orientation of an edge whose weight happens to be negative).}). With this identification, we obtain subspaces $\mathbf{V}\!_{\emph{cycle}}$ and $\mathbf{V}\!_{\emph{cocycle}}$ as the linear spans, respectively, of all basic cycles and of all basic cocycles.

For any tournament one can show that these two subspaces 
are orthogonal complements in  $\mathbb{R}^{\frac{m(m-1)}{2}}$.
  Thus each vector $\overrightarrow{w}$ $\in \mathbb{R}^{\frac{m(m-1)}{2}}$ has a unique decomposition  $\overrightarrow{w} = \overrightarrow{w}_{\emph{cycle}} + \overrightarrow{w}_{\emph{cocycle}}$
 as a sum\footnote{The electrical engineer would say that every flow of a current in a circuit can be written uniquely as a superposition of loop currents added to a superposition of sinks and sources.} in which $\overrightarrow{w}_{\emph{cycle}} \in \mathbf{V}\!_{\emph{cycle}}$, $\overrightarrow{w}_{\emph{cocycle}} \in \mathbf{V}\!_{\emph{cocycle}}$, with $\overrightarrow{w}_{\emph{cycle}} \cdot \overrightarrow{w}_{\emph{cocycle}} = 0$. See Examples 1 and 2 in this section;  \citet{DPZ} provide detailed proofs and more examples.
We say that $\overrightarrow{w}$ is \emph{purely cyclic} if $\overrightarrow{w} = \overrightarrow{w}_{\emph{cycle}}$ (with $\overrightarrow{w}_{\emph{cocycle}}= \mathbf{0}$); $\overrightarrow{w}$  is \emph{purely acyclic} if $\overrightarrow{w} = \overrightarrow{w}_{\emph{cocycle}}$ (with $\overrightarrow{w}_{\emph{cycle}}= \mathbf{0}$).  A purely cyclic $\overrightarrow{w}$ \emph{always} has majority cycles (unless it is the $\mathbf{0}$ vector), while a purely acyclic $\overrightarrow{w}$ never has any -- in fact, \emph{purely acyclic} weighted tournaments satisfy a strong, quantitative form of transitivity (see Definition \ref{QTDef}).   

\begin{example}
Consider a tournament on three vertices, with edges initially oriented cyclically: $a_1 \rightarrow a_2 \rightarrow a_3 \rightarrow a_1$. The vector space of all edge-weight assignments is then $\mathbb{R}^3$; the cyclic subspace $\mathbf{V}\!_{\emph{cycle}}$ consists of all $\overrightarrow{w}$ that assign equal weights to these three edges, forming a one-dimensional line   $L$ in $\mathbb{R}^3$ (passing through the points $(0,0,0)$ and $(1,1,1)$); and the cocyclic subspace $\mathbf{V}\!_{\emph{cocycle}}$ consists of all $\overrightarrow{w}$ that assign weights to these three edges that sum to $0$, forming a two-dimensional plane perpendicular to $L$.
\end{example}

Suppose the weight function $\overrightarrow{w}$ is of a mix of both components, with neither being zero.  In this case, the weighted tournament may or may not exhibit a majority cycle -- the situation depends on which component predominates in the sum.  If the sign of $ \overrightarrow{w}_{\emph{cycle}}+ \overrightarrow{w}_{\emph{cocycle}}$ agrees with that of $\overrightarrow{w}_{\emph{cycle}}$ on enough edges, then  $\overrightarrow{w}$ will exhibit a cycle, which we will refer to as ``overt."   If not, then stripping out the cocyclic component will necessarily reveal a cycle, which had been masked or ``hidden" by the cocyclic component.  
Thus, having no such hidden cycles is equivalent to pure acyclicity, and this condition is strictly stronger than having no overt cycles; for more on ``hidden" versus ``overt" cycles see Observation \ref{DecompObs} and footnote \ref{QTHC}.

In electrical engineering, this decomposition serves as the mathematical foundation of \emph{Kirchoff's Laws} \citep{Kirchoff} of circuit theory, but its roots lie in one-dimensional homology theory, within algebraic topology (see \citealp{Croom}; \citealp{MacWilliams};   \citealp{Harary}).  In particular, the orthogonal projection of $\overrightarrow{w}$ onto $\mathbf{V}\!_{\emph{cocycle}}$ (which yields $\overrightarrow{w}_{\emph{cocycle}}$) coincides with the \emph{boundary map} of homology.  
The decomposition was first applied to the flow of net preference by H. P. \citet{Young} in his characterization of Borda's voting rule, and was later exploited by \citet{Z} and by \citet{Saari}.  Quite recently, it was used by \citet{DPZ} to characterize the \emph{Mean Rule} for the aggregation of dichotomous weak orders, and by \citet{BBS} to characterize \emph{maximal lotteries}.  The relevance of the decomposition to profiles of weak or strict preference ballots can be appreciated from the following observation:

\begin{observation} \label{DecompObs} In the decomposition $\overrightarrow{w} = \overrightarrow{w}_{\emph{cycle}} + \overrightarrow{w}_{\emph{cocycle}}$ of any weight function $\overrightarrow{w}$ for a weighted tournament:
\begin{enumerate}
\item The cocyclic component $\overrightarrow{w}_{\emph{cocycle}}$ of $\overrightarrow{w}$ assigns to each edge $a \rightarrow b$ the scaled difference $\frac{a^\beta - b^\beta}{|\overrightarrow{V}|}$ in symmetric Borda scores\footnote{\label{SymBorda}Standard Borda scoring awards a weight of $m-1$ to the candidate ranked first on a (linearly ordered) ballot, $m-2$ to the candidate ranked second \dots and $0$ to the candidate ranked $m^\emph{th}$ (last).  The symmetric version uses weights of $m-1$, $m-3$, $m-5$, \dots , $-(m-3)$, and $-(m-1)$ (for last-ranked); equivalently, to find the weight awarded to $x$ by a ballot, subtract the number of alternatives ranked strictly over $x$ from the number ranked strictly under.  This equivalent version extends Borda scoring weights to weak orders.  The score of a candidate is the sum of weights awarded by all ballots, with candidates then ranked in order of decreasing score (and the Borda winner determined by highest score).  While standard and symmetric weights produce different scores, they induce the same final ranking.}
of the two alternatives; we might say that ``the Borda count is the boundary map.''
\item Thus, for a purely acyclic edge-weighting the pairwise majority relation yields a transitive ranking, which agrees with the ranking induced by Borda scores.  
 \item More generally, these two rankings agree when the cocyclic component is dominant in the sense that $\overrightarrow{w}(a,b)$ always agrees in sign with $\overrightarrow{w}_{\emph{cocycle}}(a,b)$.  However, when the cyclic component becomes large enough to reverse a few edge weight signs, but not so many as to introduce Condorcet cycles (``overt'' majority cycles), these rankings differ.  That difference can be attributed to the ``hidden'' cycles of the cyclic component.
\end{enumerate}
\end{observation}

\begin{example}
Consider the following profile $\mathcal{T}_{28}$ with $28$ trichotomous weak order inputs on the four-element set $\{a,b,c,d\}$:
\begin{itemize}

\item 16 copies of $\{ a\} > \{ b,c\} > \{ d\}$,

\item 8 copies of $\{ b\} > \{ c\} > \{ a,d\}$,

\item 4 copies of $\{ b,c\} > \{ d\} > \{ a\}$.

\end{itemize}

\noindent  The induced weighted tournament is on the left of Figure 2; its weight function $\overrightarrow{w}$ is a vector in $\Re ^6$.  To see why the weight on the $d \rightarrow a$ edge  is $-12$ (for example), note that the first $16$ orders have $a > d$, the next $8$ have $a \sim d$ and the last $4$ have $d > a$, for a resulting ``net preference for $d$ over $a$" of $-16 + 0 + 4 = -12$.

Symmetric Borda scores for $\mathcal{T}_{28}$ are $a^\beta = 20$, $b^\beta = 32$, $c^\beta = 16$, and $d^\beta = -68$ (see footnote \ref{SymBorda}).  As $|\overrightarrow{V}|=4$, $\frac{a^\beta - b^\beta}{|\overrightarrow{V}|} = -3 = \overrightarrow{w}_{\emph{cocycle}}(a,b)$, whence (by Observation \ref{DecompObs}.1) the $a \rightarrow b$ edge -- for example -- of the middle tournament is labeled  $-3$ in Figure 2.  Edge-weights for $\overrightarrow{w}_\emph{cycle}$ are then determined by $\overrightarrow{w}_\emph{cycle}(x,y) = \overrightarrow{w}(x,y) - \overrightarrow{w}_\emph{cocycle}(x,y)$.

The cocyclic component may now be written as the linear combination
       
\[
5a^\star + 8b^\star + 4c^\star - 17d^\star 
\]

\noindent of all four basic cocycles, in which the coefficient of $x^\star$ is $\frac{x^\beta}{|\overrightarrow{V}|} $; alternately one can use $a^\star + b^\star + c^\star + d^\star = \mathbf{0}$ to express $\overrightarrow{w}_\emph{cocycle}$ in terms of any three basic cocycles.  As shown in Figure 3, $\overrightarrow{w}_\emph{cycle}$ can be written as a linear combination

\[
7(a \rightarrow b \rightarrow d \rightarrow a) +4(b \rightarrow c \rightarrow d \rightarrow b)+3(a \rightarrow c \rightarrow d \rightarrow a)
\]

\noindent of three basic cycles.  The $-3$ on the $d\rightarrow b$ edge of $\overrightarrow{w}_\emph{cycle}$ can then be seen as the sum of $+4$ from $4(b \rightarrow c \rightarrow d \rightarrow b)$ and $-7$ from $7(a \rightarrow b \rightarrow d \rightarrow a)$ (negative because the sense of the cycle opposes that of the directed edge).

\end{example}

\begin{figure}[!ht]
{\centering
\includegraphics[height=32mm,width=95mm] {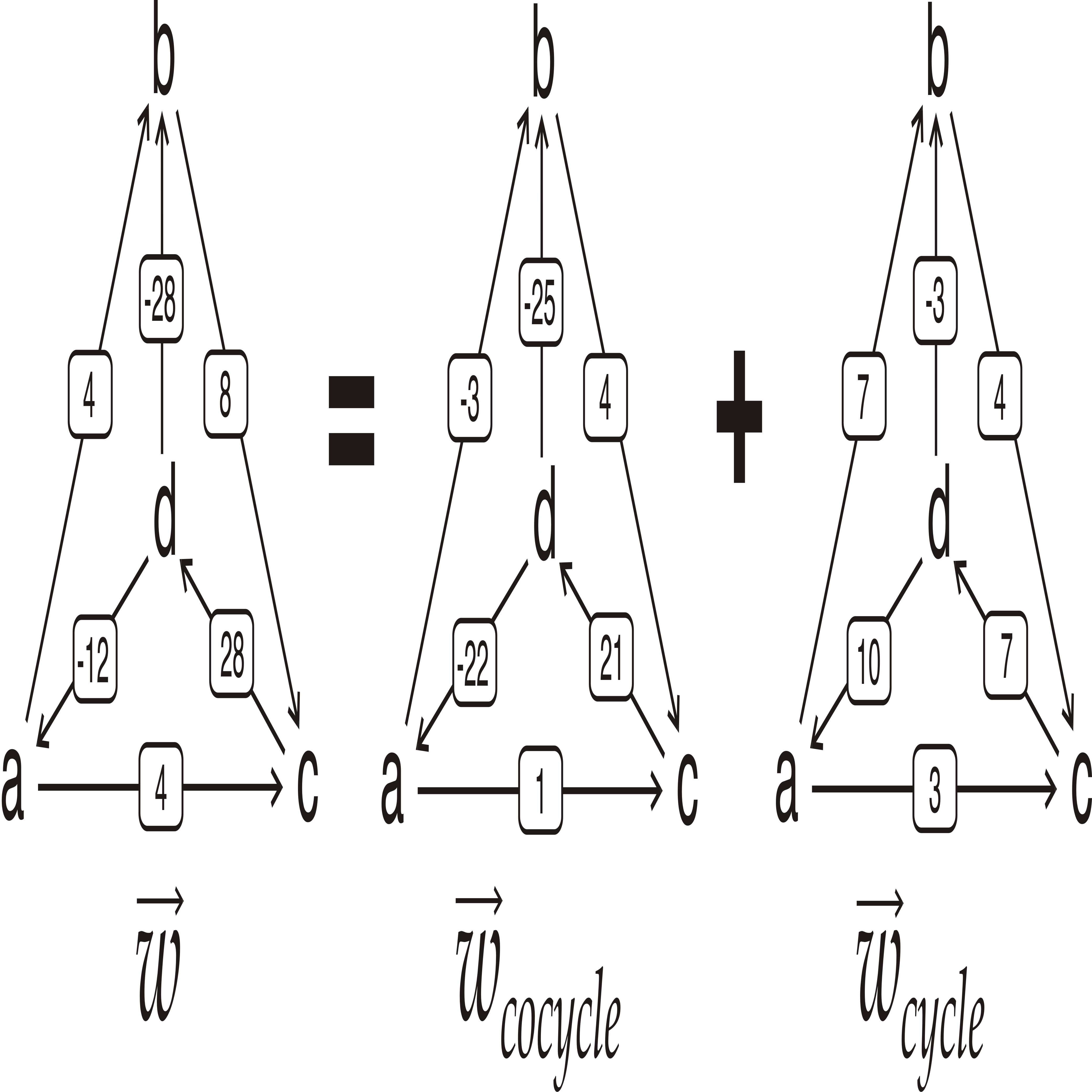}
\caption{The weighted tournament induced by profile $\mathcal{T}_{28}$, and its decomposition.}
\label{FigDecompA}}
\end{figure}

\begin{figure}[!ht]
{\centering
\includegraphics[height=32mm,width=68mm] {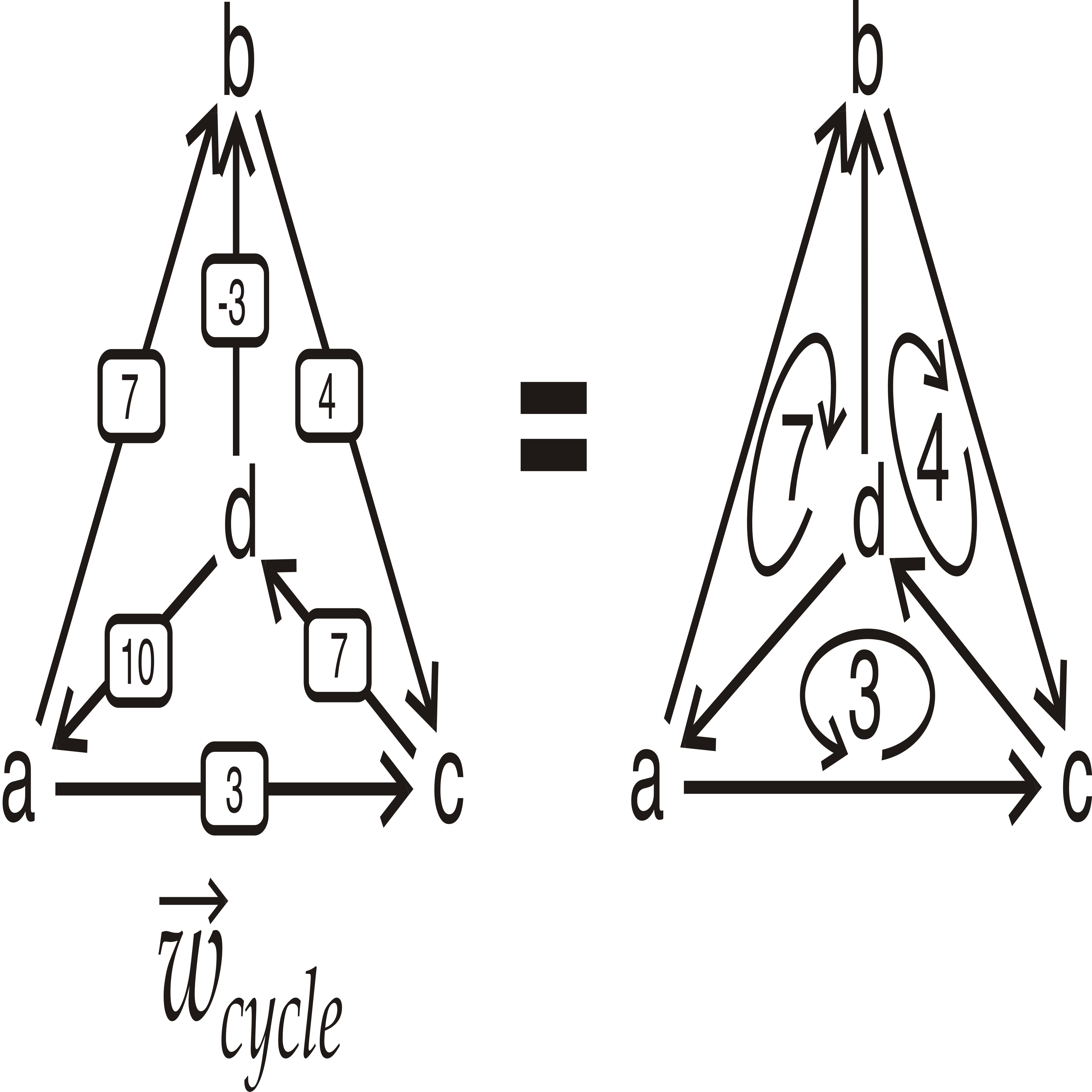}
\caption{The cyclic component from Figure 2 written as a linear combination of three basic cycles.}}
\end{figure}

\section{A Version of \emph{Max Cut} for Weighted Tournaments} \label{MaxCutSection}

In the well known \textit{max cut problem}, one starts with an undirected graph $\mathcal{G} = (V, E)$ with finite vertex set $V$.  A \emph{vertex cut} is a partition $\mathcal{P}= \{J, K\}$ of $V$ into two pieces, and is assigned a score $v(\mathcal{P})$ equal to the number of edges $\{a,b\}\in E$ whose vertices are ``cut'' by $\mathcal{P}$ (meaning $a \in J$ and $b \in K$, or $a \in K$ and $b \in J$).  The \emph{max cut} (decision) problem takes $\mathcal{G}$ along with a positive integer $s$ as input, and asks whether there exists a vertex cut of $\mathcal{G}$ with score at least $s$. It is one of the best known $\mathit{NP}$-complete problems. The corresponding optimization problem seeks a cut of maximal score; any such optimization problem inherits $\mathit{NP}$-hardness when the associated decision problem is  $\mathit{NP}$-hard.

For our purposes, a certain generalization will be useful.  A  \textit{vertex $k$-partition} $\mathcal{P}$ has $k$ nonempty pieces, with score equal to the number of edges $\{a,b\}\in E$ whose vertices belong to different pieces of $\mathcal{P}$.  The \emph{max $k$-cut problem} then asks whether there exists any $k$-partition meeting or exceeding some specified threshold score, with corresponding optimization problem formulated as one would expect. Our main concern will be with the weighted versions of these problems:  in a weighted graph $\mathcal{G} = (V, E,w)$ the function $w$ assigns an \emph{edge-weight} $w(e) \geq 0$ to each edge $e\in E$ and the \emph{score} is the sum of the weights assigned to edges that are cut.  The decision problem (or optimization problem) is similarly $\mathit{NP}$-complete (respectively, $\mathit{NP}$-hard).  For the weighted version there is no loss of generality in assuming $\mathcal{G}$ is complete; just add all the missing edges and assign them weight zero.

We consider a version of \emph{max cut} for weighted tournaments or directed graphs $\overrightarrow{\mathcal{H}} = (\overrightarrow{V}, \overrightarrow{E}, \overrightarrow{w})$ for which the real-number edge-weights may be negative,\footnote
{
For the directed problem, allowing negative weights adds no generality; if one reverses an edge while simultaneously reversing the sign of its weight, the effect on the \emph{max k-OP} problem (see Definition \ref{Def1}) is nil.  Negative weights provide the notational flexibility needed to express the decomposition $\overrightarrow{w}=\overrightarrow{w}_{\emph{cycle}} + \overrightarrow{w}_{\emph{cocycle}}$. 
} with a linearly ordered partition of the vertices playing the role of a ``cut."  For example, we might partition $\overrightarrow{V}$ into two disjoint and nonempty pieces, $T$ (for \emph{top}) and $B$ (for \emph{bottom}); the ordered partition $\overrightarrow{\mathcal{P}} = \{T>B\}$ is then equivalent to a dichotomous weak order $\succeq$ on $\overrightarrow{V}$. 
An ordered tripartition $ \{T>M>B\}$ similarly corresponds to a \emph{trichotomous} weak order on $\overrightarrow{V}$, 
while a \emph{linear} order on $V$ is equivalent to an ordered $|\overrightarrow{V}|$-partition, which has as many non-empty pieces as there are vertices, so that each piece is a singleton. 

Given a weighted tournament $\overrightarrow{\mathcal{H}} = (\overrightarrow{V}, \overrightarrow{E}, \overrightarrow{w})$ and an ordered $k$-partition $\overrightarrow{\mathcal{P}}$ corresponding to a $k$-chotomous weak order $\succeq$ on $\overrightarrow{V}$, we say that an edge $(x,y)$ goes \emph{down} if $x\succ y$, goes \emph{up} if $y \succ x$, and goes \emph{sideways} if $x \sim y$. For the example in Figure \ref{FigOne},  $(a,c)$, $(a,e)$ and $(d,f)$ go down; $(g,b)$, $(g,d)$ and $(c,b)$ go up; and $(a,b)$ and $(f,e)$ go sideways. The score $\overrightarrow{v}$ of an ordered $k$-partition $\overrightarrow{\mathcal{P}}$ is now defined from the original (not extended) $\overrightarrow{w}$ by:

\begin{equation} \label{E1}
 \overrightarrow{v}\!_ {\overrightarrow{w}}
  (\overrightarrow{\mathcal{P}} ) 
  = \hspace{-2 mm} 
 \sum_{(x,y) \in \overrightarrow{E}\;\emph{goes down}}\hspace{-9 mm}  \overrightarrow{w} (x,y) \hspace{5 mm} -  \hspace{-2.5 mm} 
 \sum_{(u,v)\in \overrightarrow{E}\;\emph{goes up}} \hspace{-7 mm}  \overrightarrow{w} (u,v). 
\end{equation}

\noindent Note that the score is independent of weights on sideways edges.  Thus in Figure \ref{FigOne} we have \begin{equation}
\overrightarrow{v}\!_{\overrightarrow{w}}(\overrightarrow{\mathcal{P}})=[3 + 4 + 5]-[1+3 +4]=4.
\end{equation}

\vspace{-2mm}

\begin{figure}[!ht]
{\centering
\includegraphics[height=44mm,width=60mm] {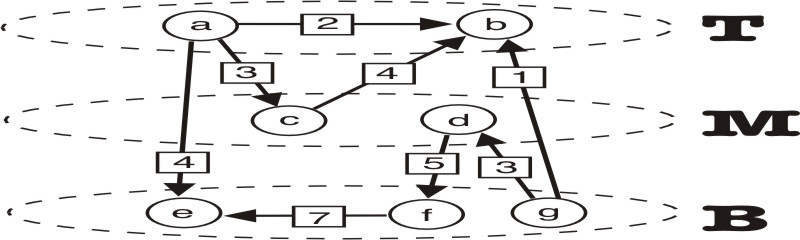}
\caption{A weighted directed graph and ordered $3$-partition}\label{FigOne}}
\end{figure}

\noindent If we instead interpret $\overrightarrow{w}$ as denoting the antisymmetric extension of the original weight function (see Section \ref{Preliminaries}),
equation (\ref{E1}) can be rewritten as:

\begin{equation}\label{E2}
 \overrightarrow{v}\!_{\overrightarrow{w}}(\overrightarrow{\mathcal{P}} ) = \sum_{x\, \succ \, y} \overrightarrow{w}(x,y). \end{equation}

\begin{definition}\label{Def1}
Max $k$-OP (``OP" for ``ordered partition") is the decision problem that takes as input a weighted tournament $\overrightarrow{\mathcal{H}}=(\overrightarrow{V}, \overrightarrow{E},\overrightarrow{w})$ for which $\overrightarrow{w}$ is \emph{integer}-valued, along with an integer $s$, and asks whether there exists an ordered $k$-partition with score at least $s$.  The corresponding optimization problem seeks an ordered partition achieving maximal score.
\end{definition}

Why propose this particular adaptation of \emph{max cut} for tournaments?   For one thing, \emph{max $k$-OP}  serves as the basis for a generalization (in Section \ref{jkSection}) of the Kemeny voting rule that yields a variety of aggregation rules -- rules already known to social choice and judgement aggregation -- as special cases. 
A second justification arises from mathematical naturality; solving \emph{max $k$-OP} is equivalent to finding a vector (from among those representing $k$-chotomous weak orders) that maximizes the inner product with a given, second vector (representing the weight function $\overrightarrow{w}$).  The equivalence will be discussed by \citet{Z2018}. 

Our goal is to show that both the purely cyclic and the transitive subcases of \emph{max $k$-OP} are $\mathit{NP}$-complete for $k \geq 3$, but polynomial-time solvable for $k=2$ and for arbitrary $k$ when $\overrightarrow{w}_\emph{cycle}=\mathbf{0}$.
A first hardness result reduces \emph{max }$3$-\emph{cut} to the purely cyclic sub-case of \emph{max }$3$-\emph{OP},  while a second construction reduces \emph{max }$2$-\emph{cut} to the transitive sub-case of \emph{max }$3$-\emph{OP}.  Either reduction alone suffices to show that \emph{max }$3$-\emph{OP} is $\mathit{NP}$-complete (as \emph{max }$3$-\emph{OP} is easily seen to lie in $\mathit{NP}$), but the two serve different roles.  The first establishes that intractability can arise from the cyclic component in isolation -- a mix of both components is not required.  However, it leaves open the possibility that overt majority cycles alone suffice for the analysis, with no role left for the orthogonal decomposition and its revelation of hidden cycles in explaining why some instances of \emph{max }$k$-\emph{OP} are in $P$ while others are $\mathit{NP}$-hard.\footnote{We are indebted to one of the referees of an earlier \emph{COMSOC} conference version of this paper, who pointed us to this issue.  The referee observed that if we assume transitivity of the majority preference relation, identifying the winning order for the standard Kemeny rule (for aggregating linear orders into a linear order) is computationally easy, and asked whether the same assumption might suffice to render the $(3,3)$-Kemeny winner problem easy.  The second reduction shows the answer to be ``no."
}  By showing that the transitive subcase remains $\mathit{NP}$-hard, the second reduction closes off this possibility and (in light of the first easiness result) establishes the essential role of hidden cycles in explaining the hardness boundary in \emph{max }$3$-\emph{OP}.

\begin{theorem}\label{Thm1}
Let $k \geq 3$.
\begin{enumerate}
\vspace{-2 mm}
\item The problem \emph{max $k$-OP} is $\mathit{NP}$-complete, even in the purely cyclic case.
\vspace{-2 mm}
\item The problem \emph{max $k$-OP} is $\mathit{NP}$-complete, even in the transitive case.
\end{enumerate}
\end{theorem}

Our first easiness result shows that the purely acyclic sub-case of \emph{max }$k$-\emph{OP} is solvable in polynomial time, for each $k \geq 2$.  The second result shows that for the special case $k = 2$, the cyclic component has no effect on the the solution to \emph{max }$k$-\emph{OP}; we may safely suppress it, at which point the first easiness result applies.

\begin{theorem}\label{Thm2}
The following problems are polynomial-time solvable:
\begin{enumerate}
\vspace{-2 mm}

\item the purely acyclic case of \emph{max $k$-OP} (for each fixed $k \geq 2$, and also with $k$ folded in to the input),
\vspace{-2 mm}

\item \emph{max $2$-OP}.
\end{enumerate}
\end{theorem}

\noindent In combination these two theorems argue that the two components have \emph{independent} effects on computational complexity.\footnote{What is the added value of \emph{orthogonality} to an orthogonal decomposition?  Consider the force of gravity on a block resting on an inclined plane,  decomposed into a first component of force normal to the plane and a second parallel to the plane.  It is precisely the orthogonality of these two components that guarantees that the first force has zero  tendency to make the block slide, while the second has zero tendency to make the block stick; the effects are independent because the components are orthogonal.}

If we attacked \emph{max $k$-OP} via brute force search over all ordered $k$-partitions (of the vertex set $\overrightarrow{V}$ of a weighted tournament), then for any fixed $k \geq 2$ we would find that the number of such partitions grows exponentially in the number $|\overrightarrow{V}|$ of vertices.  The key idea in the proof of Theorem \ref{Thm2}.1
is that this search space can be reduced to one of polynomial size $\mathcal{O}( |\overrightarrow{V}|^{k-1})$ for fixed $k$ and \emph{purely acyclic} $\overrightarrow{w} $; part 2 follows from part 1, once we show that the cyclic component may be ignored for the special case of ordered partitions into exactly two pieces. 
 
We turn now to the proof of Theorem \ref{Thm1}.1, beginning with a polynomial reduction of \emph{max $3$-cut} to \emph{max $3$-OP}.  The  idea is to replace a weighted graph $\mathcal{G}$ with a correspondingly weighted tournament  $\overrightarrow{\mathcal{H}_{\mathcal{G}}}$ in such a way that each tripartition $\mathcal{P}$ of $\mathcal{G}$'s vertices corresponds to an ordered tripartition $\overrightarrow{\mathcal{P}}$ of $\overrightarrow{\mathcal{H}_{\mathcal{G}}}$'s  vertices satisfying $v(\mathcal{P}) = \overrightarrow{v}(\overrightarrow{\mathcal{P}})$.  The $\overrightarrow{\mathcal{H}_{\mathcal{G}}}$ construction produces, for each edge $\{a,b\}$ of $\mathcal{G}$, two new vertices  (in addition to the original vertices of $\mathcal{G}$) and four edges.  More precisely:

\begin{definition}  
Let $\mathcal{G} = (V, E)$ be any complete (finite, undirected) graph and $w\!  : \! E \rightarrow \mathbb{R}$ be an associated nonnegative edge-weight function. The induced weighted tournament $\overrightarrow{\mathcal{H}_{\mathcal{G}}}$ is defined as follows:
\begin{itemize}
\item For each edge $e=\{a, b\} \in E$ of $\mathcal{G}$, construct two \emph{direction vertices} $d_{ab}$  and $d_{ba}$ of $\mathcal{H}_{\mathcal{G}}$.  Let $D = \{d_{ab}\; | \; \{a, b\} \in E\}$ denote the set of direction vertices and assume $D\cap V=\emptyset$.
\item $\overrightarrow{\mathcal{H}_{\mathcal{G}}}$'s vertex set is $\overrightarrow{V}=D \cup V$, with elements of $V$ referred to as \emph{ordinary vertices}.
\item Add all edges of form $a \rightarrow d_{ab}$ and $d_{ab} \rightarrow b$ to $\overrightarrow{\mathcal{H}_{\mathcal{G}}}$, with $\overrightarrow{w}$ assigning to each the original weight $w(\{ a, b\})$ of $\{a,b\}$ in $\mathcal{G}$; then add enough arbitrarily directed edges to make $\overrightarrow{\mathcal{H}_{\mathcal{G}}}$ a tournament, with $\overrightarrow{w}$ assigning weight $0$ to each of these.  
\end{itemize}
\end{definition}

\noindent Notice that each edge $e=\{a, b\}$ of $\mathcal{G}$ thus contributes an \emph{$\{a,b\}$  $4$\emph{-cycle}}  

\begin{equation}\label{4cycle}
 a \longrightarrow d_{ab} \longrightarrow b \longrightarrow d_{ba} \longrightarrow a  
\end{equation}

\noindent of edges in $\overrightarrow{\mathcal{H}_{\mathcal{G}}}$, each with weight $w(\{ a,b \})$.  In particular, $\overrightarrow{w}$ is \emph{purely cyclic}.  The combinatorial core of the Theorem 
 \ref{Thm1}.1 proof consists of the following:

\begin{lemma}\label{4cycleLemma} \emph{(Fitting a four-cycle into three levels)}
Let $\overrightarrow{\mathcal{P}}= \{T>M>B\}$ be any ordered tripartition of the vertex set $\overrightarrow{V}$ of $\overrightarrow{\mathcal{H}_{\mathcal{G}}}$.  Then for each weight $w$ edge $e=\{a, b\}$ of $\mathcal{G}$:  \begin{itemize}
\item if $a$ and $b$ belong to the same piece of $\overrightarrow{\mathcal{P}}$ then the net contribution to the score $\overrightarrow{v}(\overrightarrow{\mathcal{P}})$ made by the  edges of the $\{a,b\}$ $4$-cycle of line (\ref{4cycle}) is zero, and 
\item if $a$ and $b$ belong to any two different pieces of $\overrightarrow{\mathcal{P}}$ then, by appropriately reassigning the direction vertices $d_{ab}$ and $d_{ba}$ among $T$, $M$, and $B$, we can set the net contribution to $\overrightarrow{v}(\overrightarrow{\mathcal{P}})$ made by the edges of the $\{a,b\}$ $4$-cycle equal to $0$, or  $+ w$, or  $-w$, as we prefer.   
\end{itemize}
\end{lemma}

\begin{proof} [Proof]  Figures \ref{FigTwo}L, \ref{FigTwo}C, and \ref{FigTwo}R (for \emph{Left, Center, Right}) show three possible ways to assign the four vertices $a, b, d_{ab}$ and $d_{ba}$ to membership in the three pieces of $\overrightarrow{\mathcal{P}}$.  In \ref{FigTwo}R ordinary vertices $a$ and $b$ belong to the same piece (here, piece $M$) of $\overrightarrow{\mathcal{P}}$.  Of the four edges in the $\{a,b\}$ 4-cycle, two are up edges and two are down edges; as each edge has weight $w$ the net contribution of these four edges is zero.  More generally, whenever $a,b \in M$ it is easy to see that the number of up edges from the $\{a,b\}$ 4-cycle must  equal the number of down edges, no matter where $d_{ab}$ and $d_{ba}$ are placed, and that this remains true in case $a,b \in T$ or $a,b \in B$.  Thus the net contribution is $0$ whenever the ordinary vertices $a$ and $b$ are in the same piece.

\begin{figure}[!ht]
{\centering
\includegraphics[height=34mm,width=140mm] {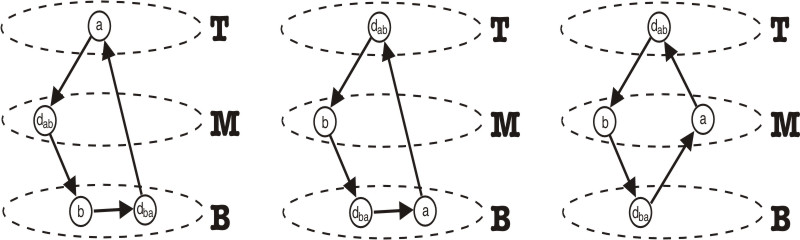}
\caption{Some possible ways to fit an $\{ a,b \}$ $4$-cycle into three levels.}\label{FigTwo}}
\end{figure}

In \ref{FigTwo}L and \ref{FigTwo}C, ordinary vertices $a$ and $b$ are in different pieces, and we have placed $d_{ab}$ and $d_{ba}$ so that there are two down edges and one up edge.  If each edge has weight $w$ then the net contribution of the four edges shown is $w$.  If we exchange the placements of $d_{ab}$ and $d_{ba}$ in \ref{FigTwo}L (or in \ref{FigTwo}C), we wind up with two up edges and one down edge for a net contribution of $-w$; if we move $d_{ab}$ and $d_{ba}$ into a common piece, then (as in the previous paragraph) the number of up edges will be equal to the number of down edges, for a net contribution of zero.  A moment's thought will convince the reader that for \emph{all} cases in which ordinary vertices $a$ and $b$ belong to different pieces, exactly four possibilities -- two up edges and one down, two down edges and one up, one down edge and one up, or two down edges and two up --    can be achieved by moving $d_{ab}$ and $d_{ba}$ around.  This completes the Lemma 1 proof.  \end{proof}

We are now ready to prove Theorem \ref{Thm1}, parts 1 and 2.

\begin{proof} [Proof of Theorem \ref{Thm1}.1]  Continuing, for the moment, with the special case $k=3$, it suffices to show that given an edge-weighted graph $\mathcal{G}$ and a positive integer $r$ the answer to the decision problem ``Does there exist a vertex tripartition $\mathcal{P}= \{J, K, L\}$ of $V$ with $v(\mathcal{P}) \geq r$?" is the same as the answer to  ``Does there exist an ordered tripartition $\overrightarrow{\mathcal{P}}$ of the vertex set $\overrightarrow{V}$ of $\overrightarrow{\mathcal{H}_{\mathcal{G}}}$ with $\overrightarrow{v}(\overrightarrow{\mathcal{P}}) \geq r$?"  

Lemma \ref{4cycleLemma} makes this easy.  Given a tripartition $\mathcal{P}= \{J, K, L\}$ of $V$ with $v(\mathcal{P}) = q \geq r$, arbitrarily order $\{J, K, L\}$ as $\{J> K>L\}$, which becomes the ordered partition of $\overrightarrow{\mathcal{H}_{\mathcal{G}}}$'s ordinary vertices.  For each weight $w$ edge $\{a,b\}$ of $\mathcal{G}$ cut by $\mathcal{P}$, add each direction vertex $d_{ab}$, $d_{ba}$ to one of the sets in $\{J> K>L\}$, so as to create two down edges and one up edge from the $\{a, b\}$ 4-cycle; for each original uncut edge $\{a,b\}$ of $\mathcal{G}$ add each vertex $d_{ab}$, $d_{ba}$ to one of the sets $\{J> K>L\}$  arbitrarily.  It is easy to see that the resulting $\overrightarrow{\mathcal{P}}$ achieves the exact same score: $\overrightarrow{v}(\overrightarrow{\mathcal{P}}) = v(\mathcal{P}) = q\geq r.$  

In the other direction, consider an ordered tripartition $\overrightarrow{\mathcal{P}}=\{V_1 \cup D_1 > V_2 \cup D_2 > V_3 \cup D_3\} $ of $\overrightarrow{V}$ with  $V_1 \cup V_2 \cup V_3 = V$,  $D_1 \cup D_2 \cup D_3 = D$, and $\overrightarrow{v}(\overrightarrow{\mathcal{P}}) \geq r$.  Let $\mathcal{P} = \{V_1, V_2, V_3\}$, a tripartition of $V$.  Each weight $w$ edge $\{a,b\}$ of $\mathcal{G}$ cut by $\mathcal{P}$ contributes $w$ to $v(\mathcal{P})$ and contributes $w$ or $0$ or $-w$ to $\overrightarrow{v}(\overrightarrow{\mathcal{P}})$.  Thus 
$v(\mathcal{P}) \geq \overrightarrow{v}(\overrightarrow{\mathcal{P}}) \geq r$, as desired.  
This completes our polynomial reduction of \emph{max $3$-cut} to \emph{max $3$-OP}.

Our argument generalizes easily to a polynomial reduction of \emph{max $k$-cut} to \emph{max $k$-OP}; in particular, the construction of $\overrightarrow{\mathcal{H}_{\mathcal{G}}}$ does not change at all, while the proof of Lemma 1 generalizes straightforwardly for ordered partitions into $k \geq 3$ pieces.

We will combine this reduction with the following straightforward reduction of \emph{max cut} to \emph{max $k$-cut}: to a weighted graph add $k-2$ new vertices, along with new edges joining them to each other and to all old vertices, and place very large weights on all these new edges.  A maximal-score partition $\mathcal{P}^\dagger$ of the enlarged weighted graph then consists of any maximal score partition $\mathcal{P}$ of the original vertices into two pieces, combined with $k-2$ other sets, each containing one new vertex.  Thus $\mathcal{P}$  may be read off from $\mathcal{P}^\dagger$. 

The combination reduces \emph{max cut} to (the purely cyclic case of)  \emph{max $k$-OP}, for each $k \geq 3$. As  \emph{max cut} is known to be $\mathit{NP}$-complete, while  \emph{max $k$-OP} (along with its purely cyclic sub-case) is clearly in class $\mathit{NP}$, Theorem \ref{Thm1}.1 now follows.  \end{proof}

\bigskip

We turn now to the proof of Theorem \ref{Thm1}.2, which reduces  \emph{max cut} to the transitive sub-case of  \emph{max $k$-OP}.
 \emph{Transitivity} here refers to the ordinary ``qualitative'' version, as expressed for a weighted tournament\footnote{We mean that for a weighted tournament induced by some profile of weak or strict rankings, this condition is equivalent to ordinary transitivity of the strict majority preference relation as defined at the end of Section \ref{Preliminaries}.  However, the construction that follows never assigns edge-weights of $0$, and in this setting all forms of transitivity of majority preference become equivalent.
} $\overrightarrow{\mathcal{H}}=(\overrightarrow{V}, \overrightarrow{E},\overrightarrow{w})$, by:

\begin{equation}  \label{OTX}
\emph{If }\overrightarrow{w}(x,y) >0\emph{ and } \overrightarrow{w}(y,z) >0\emph{ then }   \overrightarrow{w}(x,z) >0,
\end{equation}
\noindent for all $x, y \in V$ with $x \neq y$ (with  $\overrightarrow{w}$ denoting the \emph{antisymmetric extension} of Section \ref{Preliminaries}).

\begin{proof} [Proof of Theorem \ref{Thm1}.2]
Transitive  \emph{max $3$-OP} reduces to transitive  \emph{max $k$-OP} for any (finite) $k \geq 4$, as follows: add $k-3$ new vertices $D_1, \dots, D_{k-3}$ and put very large weights on all $x \rightarrow D_j$ edges for $x \in V$ and all $D_i \rightarrow D_j$ edges for $i < j$, so that any optimal ordered partition consists of an optimal ordered partition of the original vertices into three pieces, sitting above $k-3$ additional pieces, each of form $\{ D_j \}$. It suffices, then, to reduce  \emph{max cut} to transitive  \emph{max $3$-OP}.

As in the proof of Theorem \ref{Thm1}.1, we specify a polynomial translation that converts a weighted graph $\mathcal{G}=(V,E,w)$ containing $|V|$ vertices (which -- without loss of generality -- is complete and for which $w(e)$ is a nonnegative integer for 
each edge $e=\{ a, b\} \in E$) into a 
weighted tournament  $\overrightarrow{\mathcal{F}_{\mathcal{G}}}=(\overrightarrow{V}, \overrightarrow{E},\overrightarrow{w})$. This time our goals for $\overrightarrow{\mathcal{F}_{\mathcal{G}}} 
$ are a bit different. Let $C=1+\Sigma_{e \in E} w(e)$ and $\varepsilon = \frac{1}{72|V|^4}$.  Then our construction will satisfy that:
\begin{enumerate}
\item $\overrightarrow{w}$ is transitive in the sense of equation (\ref{OTX}), and
\item For each positive integer $r$ the answer to the decision problem ``Does there exist a vertex 
bipartition $\mathcal{P}= \{J, K\}$ of $V$ with $v(\mathcal{P}) \geq r$?" is the same as the answer 
to  ``Does there exist an ordered tripartition $\overrightarrow{\mathcal{P}}$ of the vertex set $
\overrightarrow{V}$ of $\overrightarrow{\mathcal{F}_{\mathcal{G}}}$ with \begin{equation}\label{rounddown}
\overrightarrow{v}
(\overrightarrow{\mathcal{P}}) \geq \langle \! \langle 3|V|C+r\rangle \! \rangle?"  
\end{equation}
\end{enumerate} 
\noindent Here $\langle \! \langle x\rangle \! \rangle $ rounds $x$ to the nearest integer. 
Given $\mathcal{G} = (V, E)$ and $w\! \! : \! E \rightarrow \Re$ as above, we now define the \emph{weighted tournament} $\overrightarrow{\mathcal{F}_{\mathcal{G}}}=(\overrightarrow{V}, \overrightarrow{E},\overrightarrow{w})$ \emph{induced by} $\mathcal{G}$ \emph{and} $w$ as follows:
\begin{itemize}
\item Choose any reference linear order $\triangleright$ of $\mathcal{G}$'s vertex set $V$.
\item Each vertex $a\in V$ contributes a quadruple of \emph{ordinary vertices} $a_1, a_2, a_3$, and $a_4$ to $\overrightarrow{V},$ along with three directed ``placement'' edges  (see Figure \ref{FigSix}):
\begin{equation}\label{4edges}
 a_1 \longrightarrow a_2 \longrightarrow a_3 \longrightarrow a_4 
\end{equation}
weighted as follows: $\overrightarrow{w}(a_2 , a_3) = 2C$ and $\overrightarrow{w}(a_1 , a_2) = C = \overrightarrow{w}(a_3 , a_4)$. 
  
\item Each edge $e=\{a, b\}$ of $\mathcal{G}$ having weight $w$  and satisfying $ a \,\triangleright \, b$  contributes two additional \emph{direction vertices} $d_{ab}$  and $d_{ba}$ to $\overrightarrow{V}$, along with four directed ``adjustment'' edges:
\begin{equation}\label{4moreedges}
 a_2 \longrightarrow d_{ab} \longrightarrow b_2\emph{; }\, \, b_3 \longrightarrow d_{ba} \longrightarrow a_3
\end{equation}
 with each such edge assigned weight $w$ (if $w>0$) or weight $\varepsilon$ (if $w=0$).
\item To make $\overrightarrow{\mathcal{F}_{\mathcal{G}}}$ into a tournament, we need to add edges between pairs of vertices that are not yet linked in either direction.  The direction of each can be chosen carefully (as detailed immediately below) in such a way that no cycles are created.  Assign weight $\varepsilon$ to each of these ``tiny'' edges.

\end{itemize}

Suppose we temporarily erase all weights on the edges of $\overrightarrow{\mathcal{F}_{\mathcal{G}}}$, and also omit all the tiny edges, retaining the placement and adjustment edges. The resulting unweighted digraph $\mathcal{E}_{\mathcal{G}}$ is acyclic 
(meaning there are no cycles that respect edge direction), thanks to the use of the reference linear order $\triangleright$; staring at Figure \ref{FigSix} can help explain what's going on here.  Any such acyclic digraph can be extended to an acyclic tournament by adding new edges one-at-a-time.\footnote{
It is easy to see that if adding edge $x \longrightarrow y$ would introduce a cycle, and adding $y \longrightarrow x$ would also do so, then there must have been a cycle in the original digraph.  The argument is essentially the same as that by \citet{HIP}, though the context there is somewhat different: alternating cycles in an undirected pregraph (a generalization of a graph wherein a vertex pair $\{a,b\}$ may be classified as an edge or a non-edge, or may be \emph{unclassified}).
}  Once these tiny edges are added and assigned positive weight $\varepsilon$, the resulting weighted tournament is transitive, as desired.  Moreover, $\varepsilon$ is small enough to guarantee that the total contribution of all weight $\varepsilon$ edges to the score of any ordered tripartition is less than $\frac{1}{2}$ in absolute value.  Thus we are safe when, for the remainder of this proof, we simultaneously ignore the presence of weight $\varepsilon$ edges and use of rounding in condition \ref{rounddown}.

\begin{figure}[!ht] 
{\centering
\includegraphics[height=35mm,width=68mm] {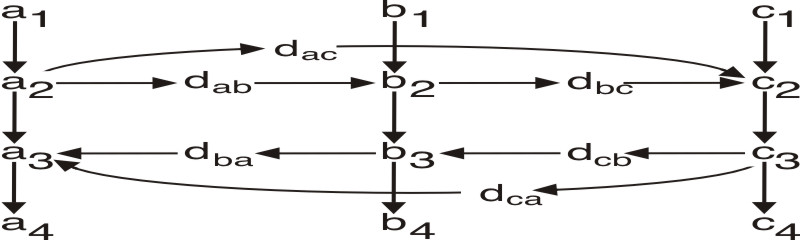}
\caption{Part of the digraph 
$\mathcal{E}_{\mathcal{G}}$
showing the (directed) placement and adjustment edges arising from three vertices 
$a \triangleright b \triangleright c$ 
(and corresponding undirected edges) of 
$\mathcal{G}$.  Note the absence of cycles.}
\label{FigSix}}
\end{figure}

Any score-maximizing tripartition $\overrightarrow{\mathcal{P}}=\{ T > M > B\}$ of $\overrightarrow{\mathcal{F}_{\mathcal{G}}}$ will also maximize that part of the score contributed 
by the placement edges (which have weight $2C$ or weight $C$), because the value of $C$ is large enough to 
make any nonzero contribution by a single weight-$C$ edge overwhelm the 
total contribution of all adjustment edges 
(which have weights from the original graph $
\mathcal{G}$).  To maximize the placement edge part of the score it is necessary, for each quadruple $a_1,  a_2, 
a_3,  a_4$ of ordinary vertices, either to place $a_1, a_2 \in T$, $a_3 \in M$, and $a_4 \in B$ -- in 
which case we will say that $\overrightarrow{\mathcal{P}}$ \emph {places $a$ up} -- or $a_1 \in T
$, $a_2 \in M$, and $a_3, a_4 \in B$ -- in which case $\overrightarrow{\mathcal{P}}$ \emph{places $a$ down}.  Either such  arrangement achieves a contribution of $3C$ from $a_1,  a_2, 
a_3,  a_4$, and it is easy to see that one cannot do better than that.  Hence, when seeking to maximize the score of $\overrightarrow{\mathcal{P}}$ we may assume that each original vertex of 
$\mathcal{G}$ is either placed up or down, for a total contribution of $3|V|C$ from placement edges. 

\begin{figure}[!ht]
{\centering
\includegraphics[height=45mm,width=100mm] {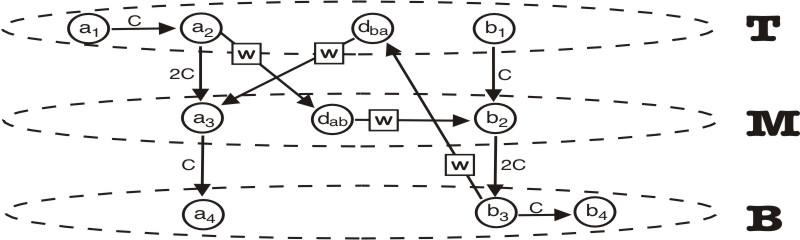}
\caption{Two quadruples (with $a$ up, $b$ down) and direction vertices (as positioned to achieve a net contribution of $w$ from the adjustment edges).}\label{FigFive}}
\end{figure}

Assume the weight of some edge $\{ a, b\}$ of $\mathcal{G}$ is $w$, and $a \triangleright b$.  Figure 5 shows the corresponding quadruples, placed in an ordered $3$-partition $\overrightarrow{\mathcal{P}}$ so that $a$ is up and $b$ down. The direction vertices $d_{ab}$ and $d_{ba}$ have been positioned so that the net contribution made to $\overrightarrow{v}(\overrightarrow{\mathcal{P}})$ by the four corresponding adjustment edges is $w+0-w+w = w$.  It is straightforward to check that one cannot achieve a contribution greater than $w$ by moving $d_{ab}$ and $d_{ba}$ into different pieces of $\overrightarrow{\mathcal{P}}$.  The situation is the same when $a$ is down and $b$ is up; one can position $d_{ab}$ and $d_{ba}$ to achieve a net contribution of $w$ from the adjustment edges, and no greater contribution is possible.  Finally, when $a$ and $b$ are either both up, or both down, the net contribution made by the four corresponding adjustment edges is $0$, regardless of where one positions $d_{ab}$ and $d_{ba}$.

Thus, given a partition $\mathcal{P} = \{ U, D\}$ of $\mathcal{G}$'s vertices with a score $v(\mathcal{P})$ of at least $r$ one can construct an ordered $3$-partition $\overrightarrow{\mathcal{P}}$ of $\overrightarrow{\mathcal{F}_{\mathcal{G}}}$'s vertices with score at least $3|V|C + r$
 by placing all vertices in $U$ up, all in $D$ down, and positioning the direction vertices to maximize the contribution of adjustment edges, as detailed in the previous paragraph.\footnote{
Recall that we are suppressing all mention of tiny edges and of $\varepsilon$ contributions to $\overrightarrow{v}(\overrightarrow{\mathcal{P}})$.
}
Conversely, given an ordered $3$-partition $\overrightarrow{\mathcal{P}}$ of $\overrightarrow{\mathcal{F}_{\mathcal{G}}}$'s vertices with score at least $3|V|C + r$, we know that each original vertex must be placed in the \emph{up} or \emph{down} position.  The considerations of the previous paragraph then imply that by setting $U = \{a \in V\, | \, a\emph{ is up}\}$ and  $V = \{b \in V\, | \, b\emph{ is down}\}$ we obtain a partition $\{ U, V\}$ of $\mathcal{G}$'s vertices with a score of at least $r$. \end{proof}

\medskip

The proof of Theorem \ref{Thm2} rests on a sequence of lemmas and definitions, including the  abstract definition of Borda score\footnote
{
In Section \ref{jkSection} we obtain a weighted tournament $\overrightarrow{\mathcal{H}}_\Pi = (A,E,\overrightarrow{w}_\Pi)$ from a profile $\Pi$ of weak (or linear) orders over a finite set $A$ of $m$ alternatives. The  score of a vertex $a \in A$ according to Definition \ref{BordaDef} coincides with the conventional notion of $a$'s Borda score based on $\Pi$, as calculated using the ``symmetric''  Borda scoring weights $m-1, m-3, \dots ,\ 3-m, 1-m$ of footnote \ref{SymBorda}. 
} 
as the ``net weighted out-degree" of a vertex $x$ in a weighted tournament.  \emph{In the following definitions $\overrightarrow{w}$ denotes the antisymmetric extension, as discussed in Section \ref{Preliminaries}.} 

\begin{definition} \label{BordaDef} Given a vertex $x$ of a weighted tournament $\overrightarrow{\mathcal{H}}=(\overrightarrow{V}, \overrightarrow{E},\overrightarrow{w})$, $x$'s Borda score is given by:
\begin{equation} \label{PlainBorda}
x^\beta = \sum_{ y\, \in V} \overrightarrow{w}(x,y)
\end{equation}
\end{definition}

\begin{definition} \label{QTDef}
A weighted tournament $\overrightarrow{\mathcal{H}}=(\overrightarrow{V}, \overrightarrow{E},\overrightarrow{w})$ satisfies \emph{exact quantitative transitivity} if 

\begin{equation}  \label{QT}
\overrightarrow{w}(x,y) + \overrightarrow{w}(y,z) = \overrightarrow{w}(x,z)
\end{equation}
holds for every three distinct vertices $x,y,z \in \overrightarrow{V}$.
\end{definition}

\begin{definition} \label{DiffGenDef} 
A weighted tournament $\overrightarrow{\mathcal{H}}=(\overrightarrow{V}, \overrightarrow{E},\overrightarrow{w})$ is \emph{difference generated} if there exists a function $\Gamma \! \! : \overrightarrow{V} \rightarrow \mathbb{R}$ such that

\begin{equation}  \label{DG}
\overrightarrow{w}(x,y) = \Gamma(x) -\Gamma(y)
\end{equation}
holds for every two distinct vertices $x,y \in \overrightarrow{V}$.  In this case, we can identify the vertices of $\overrightarrow{\mathcal{H}}$ with a sequence of real numbers
\begin{equation} \label{DGorder}
\gamma _1 \leq \gamma _2 \leq \dots \leq \gamma _m
\end{equation}
enumerating $\Gamma$'s values in non-decreasing order.
\end{definition}

We show next that pure acyclicity is equivalent to difference generated-ness as well as to quantitative transitivity:\footnote{\label {QTHC}In particular, just as transitivity rules out overt cycles, the stronger property of quantitative transitivity rules out \emph{all} cycles, hidden or overt.}
\begin{lemma} \label{charQT}
The following are equivalent for a weighted tournament $\overrightarrow{\mathcal{H}}=(\overrightarrow{V}, \overrightarrow{E},\overrightarrow{w})$:
\begin{enumerate}
\item $\overrightarrow{w}$ satisfies exact quantitative transitivity,
\item $\overrightarrow{w}$ is difference generated,
\item $\overrightarrow{w}$ is purely acyclic (equivalently, $w_{\emph{cycle}} =\mathbf{0}$ in the vector orthogonal decomposition $\overrightarrow{w} =\overrightarrow{w}_{\emph{cycle}}+\overrightarrow{w}_{\emph{cocycle}}$; equivalently,  $\overrightarrow{w} \in \mathbf{V}_{\emph{cocycle}}$, the cocycle subspace).
\end{enumerate}
\end{lemma}

\vspace{-3mm}

\begin{proof} [Proof of Lemma \ref{charQT}] For  $(1) \Rightarrow (2)$, assume $\overrightarrow{w}$ satisfies exact quantitative transitivity.  Choose any $x_0 \in \overrightarrow{V}$, and assign an arbitrary real number (a discrete analogue of a constant of integration) as $\Gamma (x_0)$'s value.  For each $y \in \overrightarrow{V}$ set $\Gamma(y) = \overrightarrow{w}(y,x) + \Gamma (x_0)$. Then for  $y,z \in \overrightarrow{V},$ $   \overrightarrow{w}(y,z) = \overrightarrow{w}(y,x) + \overrightarrow{w}(x,z) =  \overrightarrow{w}(y,x) - \overrightarrow{w}(z,x) = 
[\hspace{0.2mm} \overrightarrow{w}(y,x) + \Gamma (x_0)] - [\hspace{0.2mm} \overrightarrow{w}(z,x) + \Gamma (x_0)] =  \Gamma (y) -  \Gamma (z)$.  Conversely, if $\overrightarrow{w}$ is difference generated via $\Gamma$ then $ \overrightarrow{w}(y,x) + \overrightarrow{w}(x,z) =[\hspace{0.2mm}\Gamma(y)-\Gamma(x)] + [\hspace{0.2mm} \Gamma(x)-\Gamma(z)] = \Gamma(y)-\Gamma(z) = \overrightarrow{w}(y,z)$, so $\overrightarrow{w}$ satisfies exact quantitative transitivity.

If $\overrightarrow{w}$ is purely acyclic then, as an immediate consequence of Observation 11.2   by \citet{DPZ}, $\overrightarrow{w}$ is difference generated via the function assigning scaled symmetric Borda scores:
\begin{equation} \label{ScaledBorda}
\Gamma \!: x \mapsto \frac{x^\beta}{|\overrightarrow{V}|}.
\end{equation}
\noindent Conversely, assume $\overrightarrow{w}$ is difference generated via $\Gamma$, and let $x_1,  x_2, \dots , x_r, x_1$ be any cycle of vertices.  The corresponding \textit{basic cycle} $\sigma$ is an edge-weighting of $\overrightarrow{\mathcal{H}}$ that assigns weight one to each edge $x_i \rightarrow x_{i+1}$ or $x_r \rightarrow x_{1}$ from the vertex cycle (under the reversal convention), and weight zero to each other edge.  Thus
\begin{equation}
\overrightarrow{w} \cdot \sigma = \big{[}\Gamma(x_1)-\Gamma(x_2)\big{]} + \big{[}\Gamma(x_2)-\Gamma(x_3)\big{]} + \dots + \big{[}\Gamma(x_{r-1})-\Gamma(x_r)\big{]} + \big{[}\Gamma(x_r)-\Gamma(x_1)\big{]} =0.
\end{equation}\label{dividers}It follows from linearity of the dot product that $\overrightarrow{w} \cdot \tau =0$ holds for any linear combination of basic cycles -- hence $\overrightarrow{w} \perp \mathbf{V}_{\emph{cycle}}$, and $\overrightarrow{w} \in \mathbf{V}_{\emph{cocycle}}$.  Thus $\overrightarrow{w} $ is purely acyclic. The argument is like that for Proposition 15 by \citet{DPZ}. \end{proof}

\medskip

The central idea in showing that pure acyclicity renders  \emph{max $k$-OP} polynomial-time solvable (Theorem \ref{Thm2}.1) is that one may safely restrict attention to those ordered partitions that respect the order induced by the function $\Gamma$ whose differences generate the edge-weights: 

\begin{definition}
An ordered $k$-partition $\mathcal{P} = \{ P_k >^\prime P_{k-1} >^\prime \dots >^\prime P_1\}$ of a nondecreasing sequence $\gamma _1 \leq \gamma _2 \leq \dots \leq \gamma _m$ of real numbers is \emph{monotone} if $i < j \; \Rightarrow \pi (\gamma_i) \leq ^\prime \pi (\gamma_j)$, where $\leq ^\prime$ refers to the ordering of $\mathcal{P}$'s pieces, and $\pi (\gamma_i)$ denotes the piece $P_s$ for which $\gamma_i \in P_s$.
\end{definition}

Equivalently, monotone partitions are obtained by ``cutting'' the $\gamma$ sequence from line  (\ref{DGorder}) with $k-1$ dividers $\downarrow_i$ :
\begin{equation}\label{dividers}
\underline{\gamma _1, \gamma _2, \dots ,\gamma _{m_1}}\; \downarrow_1\; \underline{\gamma _{m_1+1}, \dots ,  \gamma _{m_2}}\;  \downarrow_2 \;\;\;  \dots \;\;\; \downarrow_{k-1} \;\underline{ \gamma _{m_{k-1}+1}, \dots ,  \gamma _{m} }
\end{equation} 

\begin{lemma} \label{MonoPart}
Given a set $S = \{ \gamma _1 \leq \gamma _2 \leq \dots \leq \gamma _m\}$ of real numbers listed in nondecreasing order let $\overrightarrow{\mathcal{H}}=(\overrightarrow{V}, \overrightarrow{E},\overrightarrow{w})$ be the weighted tournament on $\overrightarrow{V}=S$ for which $\overrightarrow{w}(\gamma _i, \gamma _j) =  \gamma _i - \gamma _j$, for each $i$, $j$ with $1 \leq i \neq j \leq m$.  There exists a monotone ordered $k$-partition of $\overrightarrow{V}$ that
  achieves maximal score (among all $k$-partitions of $\overrightarrow{V}$).\end{lemma}
\vspace{-3mm}
\begin{proof}  [Proof of Lemma \ref{MonoPart}] It is straightforward to show that if some ordered partition $\overrightarrow{\mathcal{P}}$ satisfied $i < j$ with $\pi (\gamma_i) > ^\prime \pi (\gamma_j)$ then swapping $\gamma_i$ for $\gamma_j$ (by moving $\gamma_i$ into the piece to which $\gamma_j$ initially belonged, and $\gamma_j$ into  $\gamma_i$'s initial piece) can never decrease $\overrightarrow{\mathcal{P}}$'s score. A sequence of such swaps converts $\overrightarrow{\mathcal{P}}$ into a monotone partition. \end{proof}

\begin{proof} [Proof of Theorem \ref{Thm2}.1] Given a purely acyclic instance $\overrightarrow{\mathcal{H}}=(\overrightarrow{V}, \overrightarrow{E},\overrightarrow{w})$ of  \emph{max $k$-OP} with $m=|\overrightarrow{V}|$ vertices, calculate the Borda scores from line  (\ref{PlainBorda}). Identify $\overrightarrow{V}$ with these Borda scores  $ \gamma _1 \leq \gamma _2 \leq \dots \leq \gamma _m $, arranged in nondecreasing order.  An exhaustive search would then compute  $\overrightarrow{v}(\overrightarrow{\mathcal{P}})$ for each possible monotone ordered $k$-partition of the $\gamma_j$, of which there are at most $(m-1)^{k-1}$ because there are at most $m-1$ options for placing each divider $\downarrow_i$ in line (\ref{dividers}); output any optimal monotone partition (which is an optimal partition by Lemma \ref{MonoPart}) and its score. This calculation is in $\mathcal{O}( m^{k-1}\log(mW))$ time (where $W$ is the largest weight), but with standard dynamic programming techniques this improves to $\mathcal{O}(k m^{2}\log (mW))$ time (which remains polynomial-time if $k$ is included in the input).
\end{proof}

\begin{proof}  [Proof of Theorem \ref{Thm2}.2]  For any ordered partition $\overrightarrow{\mathcal{P}}$ of a directed graph, the score $\overrightarrow{v}(\overrightarrow{\mathcal{P}})$ is a linear functional on the vector space of all possible edge weightings $\overrightarrow{w}$, so that $\overrightarrow{v}\!_{\overrightarrow{w}} (\overrightarrow{\mathcal{P}}) =  \overrightarrow{v}\!_{\overrightarrow{w}_\emph{cocycle}} (\overrightarrow{\mathcal{P}})+ \overrightarrow{v}\!_{\overrightarrow{w}_\emph{cycle}} (\overrightarrow{\mathcal{P}})$.  
Thus, once we demonstrate that for $2$-partitions $\overrightarrow{v}\!_{\overrightarrow{w}_\emph{cycle}} (\overrightarrow{\mathcal{P}})=0$, it follows that $2$-partitions also satisfy $\overrightarrow{v}\!_{\overrightarrow{w}} (\overrightarrow{\mathcal{P}})  =  \overrightarrow{v}\!_{\overrightarrow{w}_\emph{cocycle}} (\overrightarrow{\mathcal{P}})$.

Next, observe that the multiple options  for fitting a cycle into three levels of an ordered partition (Lemma \ref{4cycleLemma}, Figure \ref{FigTwo}) are severely constrained for ordered partitions having only two levels.   \noindent As suggested by the example in Figure \ref{FigThree}, for $2$-partitions the number of down edges will always equal the number of up edges.  Thus, for the basic cycle $\sigma$ that assigns weight $1$ to each edge that appears in Figure \ref{FigThree}, and weight $0$ to every edge not drawn in,  $\overrightarrow{v}\!_\sigma (\overrightarrow{\mathcal{P}}) = 0$.  By linearity, the same holds for any linear combination of basic cycles, and so we conclude that  for ordered $2$-partitions $\overrightarrow{\mathcal{P}}$, $\overrightarrow{v}\!_{\overrightarrow{w}_\emph{cycle}} (\overrightarrow{\mathcal{P}})=0$. \end{proof}
\begin{figure}[!ht]
{\centering
\includegraphics[height=39mm,width=60mm] {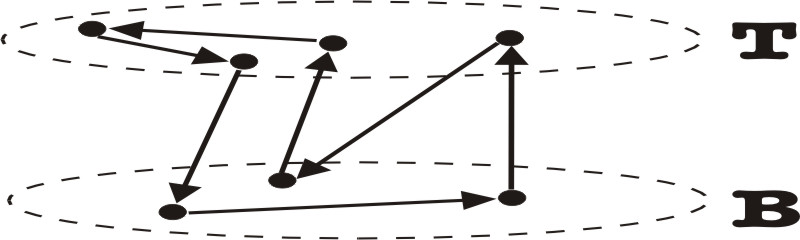}
\caption{Fitting a cycle into two levels.}\label{FigThree}}
\end{figure}

\section{The $(j,k)$-Kemeny Rule} \label{jkSection}

 Suppose that $ \{v_1, \dots , v_m\}$ is a (finite) set of $m$ alternatives, and that voters in a finite set $N$ cast weak (or linear) order ballots, resulting in a profile $\Pi = \{\geq_i\}_{i \in N}$.  The induced weighted tournament $\overrightarrow{\mathcal{H} _\Pi}
 =(\overrightarrow{V},\overrightarrow{E},\overrightarrow{w}_\Pi)$ is as follows: \vspace{-2mm}
\begin{itemize}
\item $\overrightarrow{V} =  \{v_1, \dots , v_m\}$ \vspace{-2mm}
\item $\overrightarrow{E}=\{(v_i, v_j)\, | \, i <j\}$.  Remark: This adds one edge for each two vertices. \vspace{-2mm}
\item $\overrightarrow{w}_\Pi(v_i, v_j) = |\{t\in N\, | \, v_i \geq_t v_j\}| - |\{t\in N\, | \, v_j \geq_t v_i\}|$. These weights are voters' net preferences for $v_i$ over $v_j$.
\end{itemize}

\begin{definition}\label{jkKemeny}
The \emph{$(j,k)$-Kemeny rule} takes, as input, a profile $\Pi$ of $j$-chotomous weak orders on a finite set of alternatives, and outputs the $k$-chotomous weak order(s)\footnote
{
Ties are possible.  When the number of ties is large, there may be an exponential blow-up in the number of orders in the output.  However for the first 7 rules listed in Table 2,  the output can be described in a compact language that describes a class of tied orders in terms of ties among individual alternatives. 
} on the alternatives corresponding to the solution(s) of  \emph{max $k$-OP} for  $\overrightarrow{\mathcal{H} _\Pi}$. The $(j,\mathcal{L})$-, $(\mathcal{L}, k)$- and $(\mathcal{L},\mathcal{L})$-Kemeny rules are defined similarly, with linear ordered ballots when $\mathcal{L}$ appears in the $j$ position, and linearly ordered outputs when $\mathcal{L}$ appears in the $k$ position.\footnote{
Allowing $\mathcal{L}$ as a value is useful, but represents an abuse of notation, as $j$ or $k$ then no longer serve as parameters that take only fixed numerical values.
}  A $2_r$ in either position refers to {\emph r-valent} dichotomous weak orders, corresponding to  ordered $2$-partitions $\{ T > B\}$ for which $|T| = r$.

\end{definition}

Note that a dichotomous weak order  $\{T>B\}$ can be interpreted as an approval ballot approving all alternatives in $T$. A $1$-valent (aka univalent) order $\{\{x\}>B\}$ can be interpreted as naming $x$ as winner (when it is the output) or as a plurality ballot for $x$ (when it is an input).  The $r$-valent case with $r>1$ is also of interest for representing winning committees of exogenously specified size (as in \citealp{Kilgour}; \citealp{EFSS}; \citealp{ABCEFW}).  With that understanding, approval voting will be our first example of a special case of $(j,k)$-Kemeny.  The list of such special cases contains a number of other familiar rules as well, including all in Table 2 (Section \ref{introduction}). Proofs justifying the Table 2 entries are straightforward, but (for reasons we discuss in the next section) these will appear in a later paper by \citet{Z2018}.  Our proof of Proposition \ref{AV}  for approval voting should suffice to give the flavor:

\begin{proposition} \label{AV} $(2,2_1)$-Kemeny is approval voting, with outcome the approval winner(s).
\end{proposition}

\begin{proof} [Proof of Proposition \ref{AV}] Consider a profile $\Pi^\star$ consisting of a single ballot $\{T > B\}$, corresponding to a single approval ballot of $T$, with induced weighted tournament $\overrightarrow{\mathcal{H}} _{\Pi^\star}=(\overrightarrow{V},\overrightarrow{E}, \overrightarrow{w}_{\Pi^\star})$.  It is easy to see that for any two alternatives $x$ and $y$, the weight $\overrightarrow{w}_{\Pi^\star} (x,y)$ on the $x \rightarrow y$ edge is the 
difference $\emph{App}_{\Pi^\star}(x) -  \emph{App}_{\Pi^\star}(y)$ in their approval scores (which will be $+1$, $-1$, or $0$).  Now consider a more general profile $\Pi$ with a number of ballots.  The weight $\overrightarrow{w}_{\Pi} (x,y)$ on the $x \rightarrow y$ edge of $\overrightarrow{\mathcal{H} _{\Pi}}$ is likewise the difference $\emph{App}_{\Pi}(x) -  \emph{App}_{\Pi}(y)$ in approval scores, because it is a sum of the contributions $\emph{App}_{\Pi^\star}(x) -  \emph{App}_{\Pi^\star}(y)$ made by the individual ballots.  The score of a univalent ordered partition $\{ \{x\} > \overrightarrow{V} \setminus  \{x\} \}$ will then be the sum 
\[
\Sigma_{y \in \overrightarrow{V} \setminus  \{x\} } [\emph{App}(x) - \emph{App}(y)],
\]
which is maximized when $x$ has a greatest approval score. \end{proof}

The Theorem \ref{Thm2} results now lift\footnote{
Given that the edge weights for $(j,k)$-Kemeny arise from net preferences, we can replace $W$ (in the Theorem \ref{Thm2} proof) with $n$, the number of ballots or voters, so that the final conclusion (of $\mathcal{O}(k m^{2}\log (mW))$ time) becomes $\mathcal{O}(k m^{2}\log (mn))$ time.
} immediately to $(j,k)$-Kemeny, showing:

\begin{theorem}\label{Thm4}
The problem of determining the winning ordering for a $(j,k)$-Kemeny election is in $P$ whenever $j$ or $k$ is equal to $2$ (or to $2_r$ for $r\geq 1$ any integer), and also whenever $\overrightarrow{w}_{\Pi_{\emph{cycle}}} =0$.  In particular, winner determination is in $P$ for the first seven rules of Table 2.  Also, for profiles satisfying $\overrightarrow{w}_{\Pi_{\emph{cycle}}}=0$, the  $(j, \mathcal{L})$-Kemeny outcome is the linear ranking induced by Borda scores; in particular, the original Kemeny rule agrees with Borda.
\end{theorem}

We need to be a bit more careful when lifting the $\mathit{NP}$-hardness results from Theorem \ref{Thm1} to the context of $(j,k)$-Kemeny, however.  To argue for $\mathit{NP}$-hardness when  $j \geq 3$ or  $j=\mathcal{L}$, we need to know that the specific weighted tournaments $\overrightarrow{\mathcal{H}_\mathcal{G}}$ constructed in the proofs of Theorems \ref{Thm1}.2 and \ref{Thm1}.3 are induced  by some profile $\Pi$ of $j$-chotomous orders ($j \geq 3$), and for some profile of linear orders.  Actually, from these proofs it can easily be seen that inducing some scalar multiple $C\overrightarrow{w}$ of the weights as $\overrightarrow{w}_\Pi$ would suffice, for each of these types of profile.  But given an arbitrary integer-valued $\overrightarrow{w}$ for $V = \{v_1, v_2, v_3, \dots, v_m\}$, and working for the moment with $j=3$, we can construct a profile of trichotomous weak orders satisfying $\overrightarrow{w}_\Pi = 2\overrightarrow{w}$, as follows: the profile with two trichotomous weak orders

\begin{equation}
\label{2profile}
 \Pi = \{v_1\} > \{v_2\} >\{v_3, \dots , v_m\}\; ; \; \{v_3, \dots , v_m\} > \{v_1\} > \{v_2\}
\end{equation}

\noindent satisfies $\overrightarrow{w}_\Pi (v_1,v_2) = 2$ and assigns weight $0$ to every other edge.  Thus by combining profiles similar to $\Pi$ we can build an arbitrary function $\overrightarrow{w}$ taking even integer values.  For $j>3$, modify line \ref{2profile} by  breaking $\{v_3, \dots , v_m\}$ into several pieces, ordered oppositely by the two orderings of $\Pi$. This constructions generalizes those by  \citet{McG} and by  \citet{Debord}.  We have established:

\begin{theorem}\label{Thm5}
The problem of determining the winning ordering for a $(j,k)$-Kemeny election is $\mathit{NP}$-hard if $j\geq 3$ and $k\geq 3$ both hold or if $j=\mathcal{L}$ and $k\geq 3$ both hold.  In particular winner determination is $\mathit{NP}$-hard for $(3,3)$-Kemeny.
\end{theorem}

None of our reasoning here shows $\mathit{NP}$-hardness when  $k =\mathcal{L}$; in particular,
Theorem \ref{Thm5} does not draw hardness conclusions for $(\mathcal{L},\mathcal{L})$-Kemeny (that is, for the original Kemeny rule itself) or for $(j,\mathcal{L})$-Kemeny with $j \geq 3$, because  \emph{max cut} is \emph{not} polynomially reducible to ``max $\mathcal{L}$-cut''\footnote{``Max $\mathcal{L}$-cut" is in quotes because it is silly (there exists only one unordered partition of a vertex set into singletons, so it must be the optimal such partition).  At first, it was frustrating that our methods did not seem to apply to the original Kemeny rule itself.  However, with Theorem \ref{Thm1}.2 the fundamental nature of this obstacle became clear; the methods we use here establish hardness for the transitive subcase, so they cannot possibly show hardness of Kemeny winner, which is \emph{not} hard for that case.} (unless $\mathit{P = NP}$, of course).  Nonetheless, our argument that computational complexity arises from the cyclic component also applies to the cases missing from Theorem \ref{Thm5}.  We already know \cite{BTT} that the original Kemeny rule winner problem is $\mathit{NP}$-hard, and the last clause of Theorem \ref{Thm4} tells us that Kemeny reduces to a computationally easy rule when $\overrightarrow{w}_{\Pi_{\emph{cycle}}}=0$.  As for $(j,\mathcal{L})$-Kemeny, we have just shown that for $j \geq 3$ the induced weights $\overrightarrow{w}_\Pi$ from profiles of $j$-chotomous weak orders are essentially as general as those arising from linear rankings, so winner determination is as hard as for the original Kemeny rule.

A second interesting question was raised by one of our COMSOC referees:
what happens to tractability when the cyclic component is simple? What happens, for example, if $\overrightarrow{w}_{\emph{cycle}}$ can be written as a sum of only one or two simple cycles?  We conjecture (but with low confidence) that winner determination for $(3,3)$-Kemeny would  become tractable in this case.  In this connection, it seems worth mentioning that the dimension of the cocyclic subspace grows linearly with the number of alternatives, while the dimension of the cyclic space grows quadratically.  In a sense, then, the cyclic component is inherently the more complicated one, so that sharply limiting its complexity places a rather strong restriction on the underlying profile.

\section{Connections:\! the Median Procedure, Generalized Scoring Rules,\\ \hspace{4.5 mm} and Future Work\label{Conclusions}}

The $(j,k)$-Kemeny procedure presented here is closely related to, but distinct from, the \emph{median procedure} (\citealp{BM}; \citealp{Hudry}), a better-known general method for aggregating binary relations from one class into a relation from a possibly different class.
When applied to the median procedure, several of the Table 2 restrictions yield the exact same aggregation rule as they do when applied to $(j,k)$-Kemeny; this happens, for example, when $j = \mathcal{L}$ or $k = \mathcal{L}$.  In particular, the $(\mathcal{L},\mathcal{L})$-median procedure is the Kemeny voting rule\footnote{
In this connection, social choice theorists tend to think of the median procedure as a generalization of the Kemeny voting rule, but \citet{BM} suggest the idea had parallel roots in a variety of fields.
}, the $(2,\mathcal{L})$-median procedure yields, as outcome, the ranking by approval score, and the $(\mathcal{L}, 2_1)$-median procedure yields the Borda winner(s).  

But while $(2,2)$-Kemeny yields the mean rule, the $(2,2)$-median procedure does not (it seems unlikely that any reasonable version of the mean rule arises as a median procedure restriction) and the $(2, 2_1)$-median procedure provably differs from approval voting.  So what explains the pattern of agreement and disagreement between these two general procedures?  

We will address this issue in a planned sequel \citep{Z2018}, but briefly sketch the analysis here.  The median procedure (as well as the Kemeny voting rule) is commonly defined in terms of distance minimization via an appropriate metric $\delta$, but there is an equivalent \emph{inner product} formulation that is a bit friendlier to the analysis we have in mind.  Suppose we wish to aggregate  a profile $\Pi = \{R_i\}_{i \in N}$ of binary relations belonging to some specified \emph{input class} $\mathcal{I}$, to obtain a binary relation $S$ belonging to some specified \emph{output class} $\mathcal{O}$.  We represent each binary relation $R$ (on a set $A$ with $|A| = m$) via its \emph{symmetric characteristic vector} $\mathbf{x}^R = (x^R _1, x^R _2, \dots , x^R _{m(m-1)}) $; here $x^R _{i} =1$ if the $i^{\emph{th}}$ ordered pair (in some enumeration of all ordered pairs $(a,b)$ in $A \times A$ satisfying $a \neq b$) is a member of $R$, and $x^R _{i} =-1$ otherwise.  Now each $R_i$ awards $\mathbf{x}^{R_i} \cdot \mathbf{x}^S$ points to any relation $S \in \mathcal{O}$.  Sum these points over all $i \in N$ to obtain the \emph{score} of $S$; the outcome of the aggregation is the $S \in \mathcal{O}$ with highest score.  

This method coincides with the median procedure, as usually defined.\footnote{The distance $\delta$ from $R_i$ to $S$ is related to the point award $\mathbf{x}^{R_i}\! \cdot \mathbf{x}^S$ via a negative affine transform ($\delta \mapsto A\delta +B$, where $A, B$ are real constants with $A < 0$), so that minimizing summed distance is equivalent to maximizing summed points.} Because the method awards points to binary relations rather than to the underlying alternatives in $A$, it constitutes a \emph{generalized scoring rule} (in the sense of \citealp{Myerson}; \citealp{Zwicker};  \citealp{CRX}).  

The difference between the median procedure and $(j,k)$-Kemeny can now be explained in terms of the vector space $\mathbb{R}^{m(m-1)}$ spanned by the characteristic vectors for all possible binary relations on $A$.  This space is subject to a second orthogonal decomposition, of $\mathbb{R}^{m(m-1)}$ into symmetric and antisymmetric subspaces $\mathbf{V}\!_{sym}$ and $\mathbf{V}\!_{antisym}$ (see footnote \ref{TwoEdges}).  Here $\mathbf{V}\!_{sym}$ is the span of all characteristic vectors $\mathbf{x}^{R}$ of symmetric binary relations (satisfying $(a,b) \in R \Leftrightarrow (b,a) \in R$) and $\mathbf{V}\!_{antisym}$ is the span of vectors $\mathbf{x}^{R}$ of antisymmetric binary relations (satisfying $(a,b) \in R \Leftrightarrow (b,a) \notin R$).  Let $asym(\mathbf{x}^{R}$) denote the orthogonal projection of $\mathbf{x}^{R}$ onto $\mathbf{V}\!_{antisym}$.  

We obtain an equivalent formulation of $(j,k)$-Kemeny by modifying the inner product version of the median procedure (as defined two paragraphs previous) as follows: use $asym(\mathbf{x}^{R_i} ) \cdot asym(\mathbf{x}^S )$ in place of the original $\mathbf{x}^{R_i} \! \cdot \mathbf{x}^S$.  The matter of whether some particular restriction of the median procedure agrees with the corresponding restriction for $(j,k)$-Kemeny now comes down to the role of the symmetric component in the aggregation.  If either $\mathcal{I}$ or $\mathcal{O}$ live entirely within $\mathbf{V}\!_{antisym}$, the symmetric component has no effect on the aggregation -- projection onto $\mathbf{V}\!_{antisym}$ thus has no effect on the points awarded, and the restriction of the median procedure via $\mathcal{I}$ and $\mathcal{O}$ coincides with the corresponding restriction for $(j,k)$-Kemeny.  In particular, this is what happens whenever linear rankings are used as $\mathcal{I}$ or as $\mathcal{O}$; it does not happen when $\mathcal{I}$ = $\mathcal{O}$ = the class of dichotomous weak orders.

Thus, while the proofs justifying lines of Table 2 are straightforward, we prefer to postpone these arguments to a planned sequel \citep{Z2018}, which will focus on the role of orthogonal decomposition in these sorts of procedures and in the families of specific rules that arise as their restrictions.  There we will explore more precisely the exact relationships between special cases of the median procedure, and corresponding special cases of $(j,k)$-Kemeny.

\acks{We thank Matthew Anderson, Markus Brill, Dominik Peters, and Alan D. Taylor for help with content and presentation, and the referees for an earlier (COMSOC 2016 conference) version, for suggesting interesting follow-up questions. Comments  by the journal referees significantly improved clarity and organization.}

\vskip 0.2in
\bibliography{WZjair}
\bibliographystyle{theapa}

\end{document}